%% file: main.tex
\documentclass[a4paper,USenglish,cleveref, autoref, thm-restate, authorcolumns, numberwithinsect]{lipics-v2021}
\usepackage{graphicx} 
\usepackage{tikz,cite,mathtools}
\usetikzlibrary{decorations.pathreplacing, arrows,math}
\usepackage[algo2e,vlined,linesnumbered,ruled]{algorithm2e}

\usepackage{pdfpages}
\usepackage{lscape}
\usepackage{booktabs}
\usepackage{todonotes}

\usepackage{environ}

\renewcommand{\P}{\ensuremath{\mathsf{P}}}
\newcommand{\NP}{\ensuremath{\mathsf{NP}}}
\newcommand{\cP}{\ensuremath{\mathcal{P}}}
\renewcommand{\O}{\ensuremath{\mathcal{O}}}

\newcommand{\N}{\ensuremath{\mathbb{N}}}

\newcommand{\C}{\ensuremath{\mathcal{C}}}
\newcommand{\K}{\ensuremath{\mathcal{K}}}
\newcommand{\W}{\ensuremath{\mathsf{W}}}
\newcommand{\FPT}{\ensuremath{\mathsf{FPT}}}
\newcommand{\XP}{\ensuremath{\mathsf{XP}}}

\newcommand{\tin}{\mathsf{tree} \textnormal{-} \alpha}
\newcommand{\pin}{\mathsf{path} \textnormal{-} \alpha}
\newcommand{\si}{\mathsf{si}}

\newcommand{\lmim}{\mathsf{linear} \textnormal{-} \mathsf{mim}}
\newcommand{\pw}{\mathsf{pw}}
\newcommand{\tw}{\mathsf{tw}}
\newcommand{\eccn}{\mathsf{ecc}}
\newcommand{\track}{\mathsf{t}}
\newcommand{\inumber}{\mathsf{i}}
\newcommand{\thin}{\mathsf{thin}}
\newtheorem{problem}{Problem}

\definecolor{new-blue}{RGB}{137, 171, 227}
\definecolor{new-red}{RGB}{225, 82, 61}
\definecolor{new-green}{RGB}{0, 155, 19}

\title{The Simultaneous Interval Number}
\subtitle{A New Width Parameter that Measures the Similarity to Interval Graphs}
\author{Jesse Beisegel}{Brandenburg University of Technology, Cottbus, Germany}{jesse.beisegel@b-tu.de}{https://orcid.org/0000-0002-8760-0169}{}
\author{Nina Chiarelli}{FAMNIT and IAM, University of Primorska, Koper, Slovenia}{nina.chiarelli@famnit.upr.si}{https://orcid.org/0000-0002-8169-0925}{}
\author{Ekkehard Köhler}{Brandenburg University of Technology, Cottbus, Germany}{ekkehard.koehler@b-tu.de}{}{}
\author{Martin Milani\v{c}}{FAMNIT and IAM, University of Primorska, Koper, Slovenia}{martin.milanic@upr.si}{https://orcid.org/0000-0002-8222-8097}{}
\author{Peter Mur\v{s}i\v{c}}{FAMNIT, University of Primorska, Koper, Slovenia}{peter.mursic@famnit.upr.si}{https://orcid.org/0000-0002-7350-6809}{}
\author{Robert Scheffler}{Brandenburg University of Technology, Cottbus, Germany}{robert.scheffler@b-tu.de}{https://orcid.org/0000-0001-6007-4202}{}

\titlerunning{The Simultaneous Interval Number} 

\authorrunning{J. Beisegel, N. Chiarelli, E. Köhler, M. Milani\v{c}, P. Mur\v{s}i\v{c}, R. Scheffler} 

\Copyright{Jesse Beisegel, Nina Chiarelli, Ekkehard Köhler, Martin Milani\v{c}, Peter Mur\v{s}i\v{c}, and Robert Scheffler} 

\ccsdesc[500]{Theory of computation~Parameterized complexity and exact algorithms}
\ccsdesc[500]{Theory of computation~Graph algorithms analysis}

\keywords{Interval graph, simultaneous representation, width parameter, algorithm, parameterized complexity} 

\category{} 

\acknowledgements{The authors would like to thank the reviewers of SWAT 2024 for their helpful remarks; in particular, for their hints on thinness and boxicity. The authors acknowledge partial support by the Slovenian Research and Innovation Agency (I0-0035, research programs P1-0285 and P1-0404, and research projects N1-0160, N1-0210, J1-3001, J1-3002, J1-3003, J1-4008, and J1-4084) and by the research program CogniCom (0013103) at the University of Primorska.}

\nolinenumbers 

\hideLIPIcs

\EventEditors{John Q. Open and Joan R. Access}
\EventNoEds{2}
\EventLongTitle{42nd Conference on Very Important Topics (CVIT 2016)}
\EventShortTitle{CVIT 2016}
\EventAcronym{CVIT}
\EventYear{2016}
\EventDate{December 24--27, 2016}
\EventLocation{Little Whinging, United Kingdom}
\EventLogo{}
\SeriesVolume{42}
\ArticleNo{23}

\begin{document}

\maketitle

\begin{abstract}
We propose a novel way of generalizing the class of interval graphs, via a graph width parameter called the simultaneous interval number.
This parameter is related to the simultaneous representation problem for interval graphs and defined as the smallest number $d$ of labels such that the graph admits a $d$-simultaneous interval representation, that is, an assignment of intervals and label sets to the vertices such that two vertices are adjacent if and only if the corresponding intervals, as well as their label sets, intersect.
We show that this parameter is \NP-hard to compute and give several bounds for the parameter, showing in particular that it is sandwiched between pathwidth and linear mim-width.
For classes of graphs with bounded parameter values, assuming that the graph is equipped with a simultaneous interval representation with a constant number of labels, we give \FPT{} algorithms for the clique, independent set, and dominating set problems, and hardness results for the independent dominating set and coloring problems. 
The \FPT{} results for independent set and dominating set are for the simultaneous interval number plus solution size. 
In contrast, both problems are known to be \W[1]-hard for linear mim-width plus solution size.
\end{abstract}


\input{introduction}

\input{preliminaries}

\input{sin}

\input{parameter}
\input{recognition}

\input{clique}

\input{coloring}

\input{domination}

\input{conclusion}

\bibliographystyle{plainurl}
\bibliography{lit}

\newpage
\input{appendix}

\end{document}

%% file: introduction.tex
\section{Introduction}

Interval graphs are among the best-known and most studied graph classes, due to their intuitive representation with an interval intersection model, their rich structure, and many algorithmic advantages. 
Many problems that are \NP-hard on general graphs can be solved in polynomial time on interval graphs. 
Examples are the coloring problem~\cite{gupta1979optimal,hashimoto1971wire,kernighan1973optimum}, the dominating set problem~\cite{booth1982dominating}, and the Hamiltonian cycle problem~\cite{keil1985finding}. 
Furthermore, due to their definition via interval representations, there are plenty of real-world applications for interval graphs (see~\cite{jiang2013recognizing} for a nice, short overview of such applications).

There are several different ways to generalize the concept of an interval graph. 
One of these concepts are the so-called \emph{$d$-interval graphs} where every vertex is represented by a set of $d$ intervals on the real line and two vertices are adjacent if any pair of their intervals intersect.
A subclass of these graphs are the \emph{$d$-track interval graphs} where we have $d$ parallel lines and every vertex is represented by $d$ intervals, one on each line. 
It is easy to see that any graph is a $d$-track interval graph (and, thus, a $d$-interval graph) for some $d$. Therefore, it makes sense to define the parameters \emph{interval number} $\inumber(G)$~\cite{griggs1980extremal} and \emph{track number} $\track(G)$~\cite{gyarfas1995multitrack} as the minimal numbers $d$ such that $G$ is a $d$-interval (resp. $d$-track interval) graph. 
There is some work on these graph classes concerning parameterized complexity~\cite{fellows2009parameterized,jiang2010parameterized}. 
However, most of the classical graph problems are \NP-hard for graphs with $\inumber(G) = 2$ or $\track(G) = 3$~\cite{graphclassesinterval,graphclassestrack}. 
Furthermore, even the independent set problem and the dominating set problem are $\W[1]$-hard when parameterized by the solution size for graphs with $\track(G) = 2$~\cite{jiang2010parameterized}.

Another way to define a whole family of generalizations of interval graphs comes from the so-called simultaneous representation problems. 
In this generalization, we are given $d$ interval graphs $G_1, \ldots, G_d$ which may share some vertices and asks for an interval representation that assigns to every vertex in $V(G_1) \cup \dots \cup V(G_d)$ exactly one interval such that for every $i \in \{1, \dots, d\}$ two vertices of $G_i$ are adjacent if and only if their intervals intersect. 
The problem of deciding whether a given set of graphs has such a simultaneous representation was introduced in 2009 by Jampani and Lubiw~\cite{jampani2009simconf}, where they considered chordal graphs, comparability graphs, and permutation graphs, all classes of graphs that can also be defined via certain intersection representations. 
A year later, the same authors considered the problem of simultaneous interval representations~\cite{jampani2010siminterval}. Since then, there has been several results on the complexity of this problem for different classes of graphs~\cite{bok2018note, blaesius2016simultan, rutter2019simultan}.

An equivalent definition for a simultaneous interval representation can be given as follows: For some interval model we add additional label sets in the form of subsets of $\{1, \ldots , d\}$ and two vertices belonging to two intervals are adjacent if these intervals intersect and the intersection of their label sets is non-empty.
This definition leads to an intuitive application in scheduling, where each of the labels $1, \ldots, d$ represents some machine and an interval represents a job with its processing window (the interval) and the set of machines needed to perform the job (the label set). 
An independent set in such a graph would then represent a conflict-free schedule of a subset of jobs.

Similar to $d$-interval graphs and $d$-track interval graphs, any graph can be defined as a $d$-simultaneous interval graph for some $d$. Thus, we can introduce the \emph{simultaneous interval number} $\si(G)$ as the smallest number $d$ for which $G$ is a $d$-simultaneous interval graph. 
Many width parameters are unbounded for interval graphs, as these tend to grow with the clique number (for example treewidth/pathwidth is unbounded for interval graphs). 
Furthermore, even width parameters that can be bounded for dense graphs, such as cliquewidth or twin-width, are unbounded for interval graphs~\cite{bonnet2022twin,golumbic2000clique}. 
On the other hand, those parameters that are bounded for interval graphs, such as linear mim-width or tree-independence number (see~\cite{zbMATH07191154,zbMATH07796423}), do not properly reflect the structural advantages of interval graphs. 
Many of the problems that are easy for interval graphs, such as coloring or independent set, are either para-\NP-hard or $\W[1]$-hard (see~\cref{table:comp}). 
Furthermore, the maximum clique problem is para-\NP-hard when parameterized by one of those parameters, even though the structure of the maximal cliques is very restricted for interval graphs.

When parameterized by the simultaneous interval number, however, the maximum clique problem becomes \FPT, as we will show.
In addition, some of the problems that are $\W[1]$-hard when parameterized by linear mim-width plus solution size, such as independent set and dominating set (see~\cite{fomin2020tractability,jaffke2019mimiii}), are \FPT{} when parameterized by simultaneous interval number plus solution size.
Therefore, we argue that the simultaneous interval number is a strong candidate to fill the gap in describing graphs with a structure similar to interval graphs.

\newcommand{\wh}{\W[1]}
\newcommand{\wht}{\W[2]}
\newcommand{\npc}{pNPh}


\subparagraph{Our Contribution.}
We introduce a new graph width parameter, the \emph{simultaneous interval number}, in~\Cref{sec:si-number}.
This parameter is compared to most of the other common width parameters such as treewidth, cliquewidth, or mim-width in~\Cref{sec:width-parameters}, where we also give several bounds involving the order and the size of the graph, the edge clique cover number, the clique number, and other width parameters. 
In~\Cref{sec:si-complexity} we show that the computation of the simultaneous interval number is \NP-hard. 
Furthermore, we give results on the parameterized complexity of several graph problems, such as clique (\Cref{sec:cliques}), coloring (\Cref{sec:coloring}), and variants of the independent set and dominating set problems (\Cref{sec:domination-independence}). 
For an overview of these results see~\cref{table:comp}.

\newcommand{\our}[1]{\textcolor{new-green}{#1}}

\begin{table}
    \centering
    \caption{Parameterized complexity summary. 
    Abbreviations mean ind $\to$ independent, dom $\to$ dominating, W[1] $\to$ \W[1]-hard, W[2] $\to$ \W[2]-hard, \npc{} $\to$ para-\NP-hard, $\tin$ $\to$ tree-independence number. 
    Green results are given in this paper. 
    Hardness results for problems with given solution size $k$ means that the problem is hard when parameterized by $p+k$.}\label{table:comp}
\resizebox{\textwidth}{!}{
    \begin{tabular}{l c c c c}
    \toprule
    \small
  problem\textbackslash parameter   & $p = \si(G)$ & $p = \lmim(G)$ & $p = \tin(G)$ & $p = \track(G)$  \\
  \midrule
 clique                     & $\our{\O^*(p2^{2^p+2p})}$   & \npc~\cite{vatshelle2012new} & \npc\cite{graphclasses3k1}               & \npc~\cite{francis2015maximum} \\
 clique of size $k$         & $\our{\O^*(2^{kp})}$    & ?                 & $\O^*(2^{k^p})$~\cite{chaplick2021topological,zbMATH07796423}    & $\O^*(p^kk^{k+2})$~\cite{fellows2009parameterized} \\
 coloring                   & \our{\npc}              & \npc~\cite{garey1980complexity}              & \npc~\cite{garey1980complexity}               & \npc~\cite{garey1980complexity} \\
 $k$-coloring  & $\our{\O^*(k^{kp})}$    & $\O(n^{kp})$~\cite{gonzales2024d-stable}        & $\O^*(k^{k^p})$~\cite{chaplick2021topological,zbMATH07796423}    & \npc~\cite{graphclassestrack} \\
 ind set                    & $\O(n^{p})$            & $\wh/\O(n^{2p})$~\cite{fomin2020tractability,jaffke2019mimiii}  & $\O(n^p)$~\cite{zbMATH06850324,zbMATH07796423}         & \npc~\cite{graphclassestrack}\\
 ind set of size $k$        & $\our{\O^*(2^{kp})}$    & $\wh/\O(n^{2p})$~\cite{fomin2020tractability,jaffke2019mimiii}   & ?                  & \wh~\cite{fellows2009parameterized}\\
 dom set                    & $\O(n^{2p})$            & $\wh/\O(n^{2p})$~\cite{fomin2020tractability,jaffke2019mimiii}   & \npc~\cite{bertossi1984dominating,corneil1984clustering}              & \npc~\cite{graphclassestrack} \\
 dom set of size $k$        & $\our{\O^*(2^{kp})}$    & $\wh/\O(n^{2p})$~\cite{fomin2020tractability,jaffke2019mimiii}   & \wht~\cite{liu2011parameterized}                  & \wh~\cite{fellows2009parameterized} \\
 ind dom set                & \our{\wh}/$\O(n^{2p})$  & $\wh/\O(n^{2p})$~\cite{fomin2020tractability,jaffke2019mimiii}   & ?                  & \npc~\cite{graphclassestrack} \\
 ind dom set of size $k$    & $\O(n^{2p}) $           & $\wh/\O(n^{2p})$~\cite{fomin2020tractability,jaffke2019mimiii}   & ?                  & \wh~\cite{fellows2009parameterized} \\
 \bottomrule
\end{tabular}
}
\end{table}

%% file: preliminaries.tex
\subparagraph{Definitions and Notation.}
Unless stated otherwise, all the graphs considered are simple, finite, non-empty and undirected.
Given a graph $G$, we denote by $V(G)$ its vertex set and by $E(G)$ its edge set. 
Often we will denote the number of vertices of graph, i.e., $|V(G)|$, as $n$ and the number of edges, i.e., $|E(G)|$, as $m$.
A \emph{matching} in a graph is a set of pairwise disjoint edges; a matching is \emph{induced} if no two vertices belonging to different edges of the matching are adjacent.

Next we define the term \emph{class of intersection graphs}.\label{def-intersection}
Such a graph class $\C$ can be defined via a family $S_\C$ of sets whose elements are also families of sets. 
For the sake of convenience, we assume that $S_\C$ contains a set family that contains a non-empty set.
A \emph{$\C$-representation} of a graph $G$ is a mapping $R : V(G) \to \mathcal{F}$ where $\mathcal{F} \in S_\C$ such that $xy \in E(G)$ 
if and only if $R(x) \cap R(y) \neq \emptyset$. 
We call $\mathcal{F}$ the \emph{ground set family} of $R$. 
By definition, $\C$ consists precisely of graphs $G$ having a $\C$-representation. 

The class of \emph{chordal graphs} is defined via the set $S_\C$ that contains for every tree the set of its subtrees. 
For the class of \emph{interval graphs}, the set $S_\C$ contains only the one set family, namely the set of all open intervals of the real line. 
For any interval representation $R$ of graph $G$, we define $\ell(v)$ and $r(v)$ to be the left and right endpoints of the interval $R(v)$.


A graph $G$ is a \emph{bipartite graph} if its vertex set can be partitioned into two independent sets $A$ and $B$. 
Furthermore, a bipartite graph is \emph{complete bipartite} if every vertex of $A$ is adjacent to every vertex of $B$. 
A graph is a \emph{split graph} if its vertex set can be partitioned into a clique and an independent set. 
A graph is a \emph{complete split graph} if there exists a partition in which every vertex of the independent set is adjacent to all the vertices of the clique. A graph is \emph{$C_4$-free} if it does not contain an induced cycle of length $4$.

%% file: sin.tex
\section{Simultaneous Representations and Simultaneous Interval Number}\label{sec:si-number}

In~\cite{jampani2009simconf,jampani2012simultaneous}, Jampani and Lubiw introduce the concept of \emph{simultaneous representations} as well as the \emph{simultaneous representation problem}. This concept was then taken up by Bok and Jedličková~\cite{bok2018note} who give the following definition:

\begin{definition}
Let $\mathcal{C}$ be a class of intersection graphs. Graphs $ G_1, \ldots G_{d} \in \mathcal{C}$
are \emph{simultaneously $\C$-representable} if there exist $\C$-representations $R_1, \ldots , R_d$ of $ G_1, \ldots G_d$ 
with a common ground set family $\mathcal{F} \in S_\C$ such that 
\[\forall i,j \in \{1, \dots, d\}, \forall v \in V(G_i) \cap V(G_j)\colon R_i(v) = R_j (v).\] 
In particular, we say that $G= G_1 \cup \ldots \cup G_{d} $ is a \emph{$d$-simultaneous $\C$-graph.}
\end{definition}

For convenience of notation, we will oftentimes use the following equivalent definition of a simultaneous representation.

\begin{definition}\label{def:labeling}
Let $d\in \mathbb{N}$, let $G$ be a graph, and let $L: V(G) \rightarrow \mathcal{P}(\{1, \ldots , d\})$ be a labeling of the vertices of $G$. 
Furthermore, let $G' \in \mathcal{C}$ with $V(G) = V(G')$ and  $E(G) \subseteq E(G')$ be a graph with a $\C$-representation $R$. 
We say that $(R,L)$ is a \emph{$d$-simultaneous $\C$-representation} of $G$ if it holds that $vw \in E(G)$ if and only if $ R(v) \cap R(w) \neq \emptyset $ and $ L(v) \cap L(w) \neq \emptyset $.
\end{definition}

Note that this definition allows the \emph{empty set} as a label set. Obviously, any vertex with an empty label set is isolated. 
Therefore, the graphs admitting a $0$-simultaneous $\C$-representation are exactly the edgeless graphs.

\begin{observation}\label{obs:labeling}
    Let $\C$ be a class of intersection graphs. 
    Let the graphs $ G_1, \dots, G_d \in \C$ be simultaneously $\C$-representable with $\C$-representations $R_1,\dots,R_d$ with a common ground set family $\mathcal F$. 
    Let $G\coloneqq G_1 \cup \dots \cup G_d$ and let $R\colon V(G) \to \mathcal F$ be defined as $R(v) \coloneqq R_i(v)$ for any $i$ with $v \in V(G_i)$. 
    Let $L$ be the labeling given by $L(v) = \{ i\colon  v \in G_i\}$ for all $v\in V(G)$. 
    Then $(R,L)$ is a $d$-simultaneous $\C$-representation of $G$. 
\end{observation}

\begin{sloppypar}
This observation implies that every $d$-simultaneous $\C$-graph has a $d$-simultaneous \hbox{$\C$-representation}. However, the converse is not true in general.\footnote{As an example, we consider the class $\K$ of complete graphs which can be represented as intersection graphs via the set $S_\K = \{\{1\}\}$. 
The $n$-vertex edgeless graph has a $1$-simultaneous $\K$-representation where all vertices are labeled with the empty set.
However, it is not a $d$-simultaneous $\K$-graph for any $d$. \label{fn:clique}} However, if we exclude empty label sets and unused labels, then there is an analogous result to \cref{obs:labeling}. 
\end{sloppypar}

\begin{observation}\label{obs:labeling-2}
    Let $(R,L)$ be a $d$-simultaneous $\C$-representation of a graph $G$ with $L(v) \neq \emptyset$ for all $v \in V(G)$ and such that for all $i \in \{1,\dots,d\}$ there exists a vertex $v$ with $i \in L(v)$. 
    Let $G_i$ be the subgraph of $G$ induced by the vertex set $\{v\colon i \in L(v)\}$ and let $R_i$ be the restriction of $R$ to $V(G_i)$. 
    Then the graphs $G_1,\dots,G_d$ are simultaneously $\C$-representable with $\C$-representations $R_1,\dots,R_d$.
\end{observation}

A vertex with an empty label set would have to be considered as a vertex that is in none of the graphs of a simultaneous representation. 
However, this technical addition to the definition is very useful to address the issue of isolated vertices and leads to more compact statements and simpler proofs. 
For all of the classes considered here, it is always possible to represent isolated vertices without the empty label set. 
For example, for interval graphs we can always represent such a vertex with an interval that intersects nothing else. 
However, in general we cannot assume that this is possible for any class of intersection graphs (see \cref{fn:clique}).

\begin{theorem}\label{thm:label-edges}
For every class of intersection graphs $\C$, every graph $G$ has an $|E(G)|$-simultaneous $\C$-representation.
\end{theorem}

\begin{proof}
    Let $E = \{e_1, \ldots, e_m\}$. 
    By our assumption of classes of intersection graphs, there exists a set family $\mathcal{F} \in S_\C$ that contains some non-empty set $S$. 
    For all vertices $v \in V(G)$, we define $R(v) = S$. 
    Furthermore, let $L(v) \coloneqq \{i\colon v \text{ is an endpoint of  } e_i\}$. It follows directly that two vertices of $G$ are adjacent if and only if their representations and their label sets have a non-empty intersection.
\end{proof}

In particular, this theorem holds for the class of intervals graphs, motivating the following definition.

\begin{definition}
    Let $G$ be a graph. 
    The \emph{simultaneous interval number} $\si(G)$ of $G$ is the smallest integer $d$ such that there exists a $d$-simultaneous interval representation of $G$. 
\end{definition}

As observed before, the graphs with simultaneous interval number 0 are exactly the edgeless graphs. 
Furthermore, the graphs with simultaneous interval number at most 1 are exactly the interval graphs, and 
the class of graphs with the simultaneous interval number equal to 2 contains some asteroidal triples and all cycles (see \cref{fig:si-examples}). 

\begin{figure}
    \centering
    \begin{tikzpicture}[vertex/.style={draw, circle}, int/.style={line width=1.4}, scale=0.9]
    \begin{scope}
        \node[vertex,inner sep=2pt] (1) at (0,0) {};
        \node[vertex,inner sep=2pt] (2) at (0,0.5) {};
        \node[vertex,inner sep=2pt] (3) at (0,1) {};
        \node[vertex,inner sep=2pt] (4) at (-0.354,-0.354) {};
        \node[vertex,inner sep=2pt] (5) at (-0.707,-0.707) {};
        \node[vertex,inner sep=2pt] (6) at (0.354,-0.354) {};
        \node[vertex,inner sep=2pt] (7) at (0.707,-0.707) {};

        \draw (3) -- (2) -- (1) -- (4) -- (5);
        \draw (7) -- (6) -- (1);
    \end{scope}

    \begin{scope}[xshift=3cm]
        \draw[int, lipicsYellow] (0,0) -- (3,0);
        \draw[int, lipicsYellow] (-0.5,-0.25) -- (0.5,-0.25);
        \draw[int, lipicsYellow] (-1.25,-0.5) -- (-0.25,-0.5);

        \draw[int, black] (1,0.25) -- (2,0.25);
        \draw[int, new-blue] (1,0.5) -- (2,0.5);

        \draw[int, lipicsYellow] (2.5,-0.25) -- (3.5,-0.25);
        \draw[int, lipicsYellow] (4.25,-0.5) -- (3.25,-0.5);
    \end{scope}

    \begin{scope}[xshift=9cm, yshift=-0.5cm]
        \node[vertex,inner sep=2pt] (1) at (0,0) {};
        \node[vertex,inner sep=2pt] (2) at (-0.5,0.5) {};
        \node[vertex,inner sep=2pt] (3) at (0,1) {};
        \node[vertex,inner sep=2pt] (4) at (0.5,0.5) {};
      
        \draw (1) -- (2) -- (3) -- (4) -- (1);
    \end{scope}

    \begin{scope}[xshift=11cm, yshift=-0.5cm]
        \draw[int, lipicsYellow] (0,0) -- (1,0);
        \draw[int, new-blue] (0,1) -- (1,1);
        \draw[int, black] (-0.75,0.5) -- (0.25,0.5);
        \draw[int, black] (0.75,0.5) -- (1.75,0.5);
    \end{scope}
    \end{tikzpicture}
    \caption{Two forbidden induced subgraphs of interval graphs with $2$-simultaneous interval representations. 
    Yellow intervals have label set $\{1\}$, blue intervals have label set $\{2\}$ and black intervals have label set $\{1,2\}$. Note that the representation of the 4-cycle can be extended to a $2$-simultaneous interval representation of cycles of arbitrary length.}
    \label{fig:si-examples}
\end{figure}
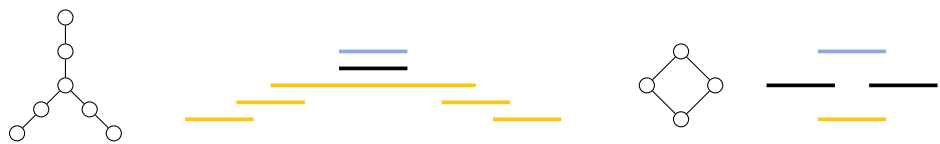

In the following, we show some bounds on the simultaneous interval number. The first result is implied directly by \cref{thm:label-edges}.

\begin{corollary}
    For any graph $G$ it holds that $\si(G) \leq |E(G)|$.
\end{corollary}

Next we show that this bound is tight, up to a constant factor. 

\begin{theorem}\label{thm:3-partite}
    Let $G$ be a complete $3$-partite graph with parts of equal size.
    Then, $\si(G) = \frac19|V(G)|^2 = \frac13 |E(G)|$.
\end{theorem}

\begin{proof}
    Let $n = |V(G)|$. 
    First we show that $\si(G) \geq \frac19n^2$. 
    Consider an arbitrary simultaneous interval representation of $(R,L)$ of $G$.
    Assume that there are two intervals that do not intersect. 
    Then these intervals belong to vertices of the same partition set. This implies that the intervals of the vertices of the other two partition sets have to intersect both intervals. 
    Therefore, they all share a common point. 
    Thus, in any case there are at least two of the three partition sets whose intervals all share a common point. 
    Let $X$ and $Y$ be two such partition sets. 
    As the vertices of a partition set form an independent set, their label sets have to be disjoint. 
    Furthermore, the label sets of all vertices of $X$ have to contain at least one label of all the label sets of the vertices in $Y$. 
    As the label sets of $Y$ are also disjoint, each label set of $X$ has at least $\frac n3$ elements. 
    Thus, in total there have at least $\frac n3 \cdot \frac n3 = \frac19 n^2$ different labels.

    Now we construct a simultaneous interval representation of $G$ with exactly $\frac19 n^2$ labels. 
    Let $X$, $Y$, and $Z$ be the three partition sets of $G$. 
    We consider $\frac n3$ pairwise disjoint sets $S_1, \ldots, S_{\frac{n}{3}}$ and enumerate their elements as $S_i = \{s_1^i, \ldots, s_{\frac{n}{3}}^i\}$. 
    Note that $|\bigcup_i S_i| = \frac19 n^2$. For all $x \in X$ let $L(x) = \bigcup_i S_i$. 
    Let $Y = \{y_1, \ldots, y_{\frac{n}{3}}\}$ and $Z = \{z_1, \ldots, z_{\frac{n}{3}}\}$. We set $L(y_i) = S_i$. 
    Furthermore, we define $L(z_i) = \{s_i^j\colon 1 \leq j \leq \frac{n}{3}\}$ for all $i\in \{1,\ldots \frac n3\}$. 
    We give all vertices of $Y$ and $Z$ the same interval and all the vertices of $X$ get pairwise disjoint intervals that are contained in the intervals given to vertices in $Y$ and $Z$. 
    It is easy to see that both the intervals and the label sets of two vertices have a non-empty intersection if and only if the vertices belong to different partition sets. 
\end{proof}

\begin{theorem}\label{thm:si-bipartite}
    Let $G = (V, E)$ be a bipartite graph with a bipartition $V = X \dot\cup Y$. 
    Then $\si(G) \leq \min\{|X|,|Y|\}$. This bound is tight for complete bipartite graphs.
\end{theorem}

\begin{proof}
    We may assume without loss of generality that $|X| \leq |Y|$. 
    Let $X = \{x_1, \ldots, x_d\}$. 
    We build a $d$-simultaneous interval representation of $(R,L)$ of $G$ as follows: 
    All the vertices of $X$ get the same interval $(0,1)$. 
    All the vertices of $Y$ get pairwise disjoint intervals that are completely contained in the interval $(0,1)$.
    Vertex $x_i \in X$ gets label set $L(x_i) = \{i\}$. 
    The label set of a vertex $y \in Y$ is chosen so that it holds $i \in L(y)$ if and only if $x_iy \in E$. 
    It follows immediately that two vertices are adjacent in $G$ if and only if both their intervals and their label sets have a non-empty intersection.

    Now consider a complete bipartite graph $G$ with a bipartition $\{X,Y\}$ where $|X| \leq |Y|$ and assume for a contradiction that $G$ admits a $d$-simultaneous interval representation of $(R,L)$ of $G$ with $d < |X|$. 
    If $|X| \le 1$, then $d = 0$ and hence $G$ is edgeless, implying $|X| = 0$, which contradicts  $d < |X|$.
    Thus, $|X|\ge 2$ and every vertex of $X$ is labeled with a nonempty label set. 
    Since $X$ is an independent set and the maximum number of pairwise disjoint label sets used on $X$ is $d$, there must exist two vertices in $X$ whose intervals are disjoint. 
    However, this implies that the intervals of all vertices in $Y$ are pairwise intersecting. 
    Since there are fewer than $|Y|$ labels, at least one label appears in the label set of two vertices of $Y$, a contradiction.
\end{proof}

The \emph{complement of a matching} is a graph obtained from a complete graph of even order~$n$ by removing from it $\frac n2$ pairwise disjoint edges.

\begin{lemma}\label{lem:complement-matching-1}
    If $G$ is the complement of a matching with $n$ vertices, then $\si(G) \geq \log_2(n-1)$.
\end{lemma}

\begin{proof}
      Consider an arbitrary simultaneous interval representation of $(R,L)$ of $G$. 
      If there exist two intervals $R(x)$ and $R(y)$ that do not intersect, then the vertices $x$ and $y$ are not adjacent. 
      Since all the other vertices must be adjacent to both $x$ and $y$ all their intervals have to pairwise intersect. 
      Furthermore, since all vertices have different neighborhoods, it follows immediately that except for $x$ and $y$ no pair of vertices can have the same label set. 
      Thus, we need at least $n-1$ different label sets. 
      This implies that the number of labels is at least $\log_2(n-1)$. 
\end{proof}

We will see later, in \cref{lemma:matching-labels}, that this bound is tight.

%% file: parameter.tex
\section[Placing si(G) in the Zoo of Graph Width Parameters]{Placing $\si(G)$ in the Zoo of Graph Width Parameters}\label{sec:width-parameters}

In this section we compare the simultaneous interval number to several other graph width parameters. 
See \cref{fig:parameters} for an overview. 
A verification of the figure can be found in \cref{tab:my_label} on p.~\pageref{tab:my_label}.

\begin{figure}
\centering
\begin{tikzpicture}[xscale=2.5,yscale=1]
  \tikzset{box/.style={draw,rectangle}}
  \tikzset{roundedbox/.style={draw,rectangle,rounded corners}}
  \tikzset{markedbox/.style={draw,rectangle,very thick, fill=lipicsYellow}}
    \footnotesize
  \node[roundedbox] (pw) at (0.75,0.75) {\mathstrut pathwidth};
  \node[roundedbox] (ecc) at (2,0.75) {\mathstrut edge clique cover number};
  \node[roundedbox] (tw) at (1.25,2.25) {\mathstrut treewidth};
  \node[markedbox] (si) at (0,1.6) {\mathstrut simultaneous interval number};
  \node[roundedbox] (palpha) at (-1.4,2.8) {\mathstrut path-independence number};
  \node[roundedbox] (talpha) at (-1.4,3.75) {\mathstrut tree-independence number};
  \node[roundedbox] (cw) at (1.85,3) {\mathstrut cliquewidth};
  \node[roundedbox] (rw) at (2.85,3) {\mathstrut rank-width};
  \node[roundedbox] (twinw) at (1.85,3.85) {\mathstrut twin-width};
  \node[roundedbox] (lmimw) at (0,3.8) {\mathstrut linear mim-width};
  \node[roundedbox] (mimw) at (0,4.75) {\mathstrut mim-width};
  \node[roundedbox] (omimw) at (0,5.75) {\mathstrut o-mim-width};
  \node[roundedbox] (simw) at (0,6.75) {\mathstrut sim-width};
  \node[roundedbox] (tn) at (1.55,5.75) {\mathstrut track number};
  \node[roundedbox] (in) at (2.7,5.75) {\mathstrut interval number};
  \node[roundedbox] (thin) at (0,2.5) {\mathstrut thinness};
  \node[roundedbox] (box) at (0.9,4.75) {\mathstrut boxicity};
 
  \draw[thick,-stealth',font=\sffamily] (pw.20) to (tw);
  \draw[thick,-stealth',font=\sffamily] (pw) to (si.-10);
  \draw[thick,-stealth',font=\sffamily] (ecc) to (si.-5);
  \draw[thick,-stealth',font=\sffamily] (ecc.143) to (cw);
  \draw[thick,-stealth',font=\sffamily] (si.170) to (palpha);
  \draw[thick,-stealth',font=\sffamily] (si) to (thin);
  \draw[thick,-stealth',font=\sffamily] (thin) to (lmimw);
  \draw[thick,-stealth',font=\sffamily] (si.10) to (tn);
      \draw[thick,stealth'-stealth',font=\sffamily] (tn) to (in);
  \draw[thick,-stealth',font=\sffamily] (lmimw) to (mimw);
  \draw[thick,-stealth',font=\sffamily] (palpha) to (talpha);
  \draw[thick,-stealth',font=\sffamily] (talpha) to (omimw);
  \draw[thick,-stealth',font=\sffamily] (tw) to (cw);
  \draw[thick,-stealth',font=\sffamily] (tw) to (talpha.-7);
  \draw[thick,-stealth',font=\sffamily] (tw) to (tn.-60);
  \draw[thick,stealth'-stealth',font=\sffamily] (cw) to (rw);
  \draw[thick,-stealth',font=\sffamily] (cw) to (twinw);
  \draw[thick,-stealth',font=\sffamily] (cw) to (mimw);
  \draw[thick,-stealth',font=\sffamily] (mimw) to (omimw);
  \draw[thick,-stealth',font=\sffamily] (omimw) to (simw);
  \draw[thick,-stealth',font=\sffamily] (thin.25) to (box.-95);
  \draw[thick,-stealth',font=\sffamily] (tw) to (box.-60);

\end{tikzpicture}
\caption{Diagram illustrating the relations between different graph width parameters. A directed edge from parameter $P$ to parameter $Q$ means that a bounded value of $P$ implies a bounded value for $Q$. If a directed path from $P$ to $Q$ is missing, then parameter $Q$ is unbounded for the graphs of bounded $P$.}\label{fig:parameters}
\end{figure}
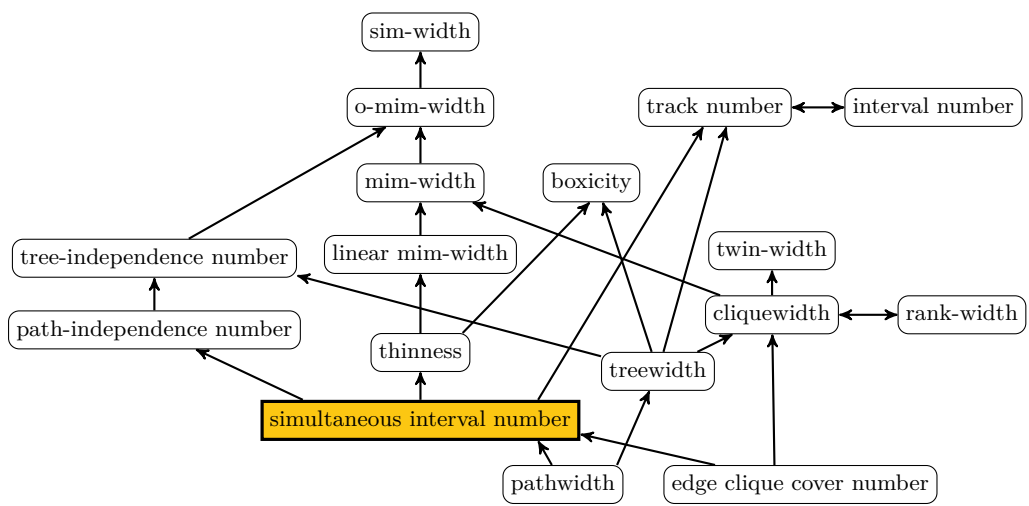


\subsection{Lower Bounds}

It is easy to see that $d$-simultaneous interval graphs are $d$-track interval graphs. 
This implies the following result.

\begin{theorem}\label{thm:track}
    Every graph satisfies $\track(G) \leq \si(G)$.
\end{theorem}

The concept of \emph{thinness} was introduced by Mannino et al.~\cite{mannino2007stable}.

\begin{definition}[Thinness]
The \emph{thinness} $\thin(G)$ of a graph $G$ is the smallest integer $k$ such that there is a partition $\{V_1, \dots, V_k\}$ of $V(G)$ and a vertex ordering $(v_1, \dots, v_n)$ of $G$ fulfilling that for any three vertices $v_a, v_b, v_c$ with $a < b < c$ and $v_a, v_b \in V_i$ for some $i$ it holds that $v_bv_c \in E(G)$ if $v_av_c \in E(G)$.
\end{definition}

\begin{theorem}\label{thm:thin}
For any graph $G$ it holds that $\thin(G) \leq 2^{\si(G)}$.
\end{theorem}

\begin{proof}
    Consider a $\si(G)$-simultaneous interval representation $(R,L)$ of $G$. 
    We define a partition of $V(G)$ as follows. 
    Two vertices $v_i, v_j$ are in the same partition set if and only if $L(v_i) = L(v_j)$. 
    Hence, there are at most $2^{\si(G)}$ sets in the partition.
    Let $\sigma = (v_1, \ldots, v_n)$ be a vertex ordering of $G$ such that for all $i,j \in \{1,\dots,n\}$ with $i < j$ it holds that $r(v_i) \leq r(v_j)$. 
    Let $v_a, v_b, v_c$ be three vertices with $a < b < c$ where $L(v_a) = L(v_b)$. 
    If $v_a$ is adjacent to $v_c$, then $r(v_a) > \ell(v_c)$ and $L(v_a) \cap L(v_c) \neq \emptyset$. 
    By construction of $\sigma$, it holds $\ell(v_c) < r(v_a) \leq r(v_b) \leq r(v_c)$ and, thus, the intervals of $v_b$ and $v_c$ intersect. 
    Furthermore, $L(v_b) \cap L(v_c) \neq \emptyset$ since $L(v_a) = L(v_c)$. 
    Hence $v_bv_c \in E(G)$.
\end{proof}
    
Complements of matchings with $n$ edges have thinness $n$~\cite{chandran2007independent}. 
We will later see in \cref{lemma:matching-labels} that the simultaneous interval number of such a graph is $\O(\log n)$. 
This implies that the bound given in \cref{thm:thin} is asymptotically sharp.  
Bipartite permutation graphs and, hence, also complete bipartite graphs have thinness at most~2~\cite{bonomo2023intersection}. 
As we have seen in \cref{thm:si-bipartite}, the simultaneous interval number of complete bipartite graphs is unbounded. 
Therefore, this class shows that bounded thinness does not imply bounded simultaneous interval number.

The concept of a linearized version of mim-width was introduced by Vatshelle~\cite{vatshelle2012new} as mim-width using a caterpillar decomposition. 
This concept has since been called \emph{linear mim-width} (for example by Golovach et al.~\cite{golovach2015output}).

\begin{sloppypar}
\begin{definition}[Linear mim-width]
    Given a graph $G$ and a vertex ordering $\sigma = (v_1, \dots, v_n)$ of $G$, we define the quantity $\lmim(G, \sigma, i)$ for $1 \leq i \leq n$ to be the maximum size of an induced matching in the bipartite graph that contains all the edges of $G$ between the two sets $\{v_1, \dots, v_i\}$ and $\{v_{i+1}, \dots, v_n\}$. 
    We define $\lmim(G, \sigma) := \max_{i \in \{1,\dots,n\}} \lmim(G, \sigma, i)$. 
    The \emph{linear mim-width} of $G$, denoted $\lmim(G)$, is defined as the minimum value $\lmim(G, \sigma)$ among all vertex orderings $\sigma$ of $G$.
\end{definition}
\end{sloppypar}

It was shown by Bonomo and de Estrada~\cite{bonomo2019thinness} that for any graph $G$ it holds that $\lmim(G) \leq \thin(G)$. 
Combining this with \cref{thm:thin} we see that bounded simultaneous interval number also implies bounded linear mim-width. 
Moreover, using a more direct argumentation, the lower bound on the simultaneous interval number given by the logarithm of the linear mim-width can be improved to a linear lower bound.

\begin{theorem}\label{thm:lmim}
    For any graph $G$ it holds that $\lmim(G)\le \si(G)$.
\end{theorem}

\begin{proof}
    Consider a $\si(G)$-simultaneous interval representation $(R,L)$ of $G$. 
    Let $(v_1, \ldots, v_n)$ be a vertex ordering of $G$ such that for all $i,j \in \{1,\dots,n\}$ with $i < j$ it holds that $r(v_i) \leq r(v_j)$. Let $i \in \{1,\dots,n\}$ be arbitrary and consider the spanning bipartite subgraph $B$ of $G$ containing the edges of $G$ between the vertex sets $\{v_1, \dots, v_i\}$ and $\{v_{i+1}, \dots, v_n\}$. Let $M$ be a maximum induced matching of $B$. For every edge $v_av_b$ in $M$ we define $\lambda(v_av_b) := L(v_a) \cap L(v_b)$. Note that $\lambda(v_av_b) \neq \emptyset$. 
    Consider two edges $v_av_b$ and $v_cv_d$ in $M$. 
    We may assume without loss of generality that $a \leq c \leq i$. 
    We know that $\ell(v_b) < r(v_a) \leq r(v_c)\leq r(v_b)$ and, thus, the intervals of $v_b$ and $v_c$ have a non-empty intersection. 
    As $M$ is induced, the edge $v_bv_c$ is not present in the graph $B$ and, thus, not part of $G$. 
    Therefore, $L(v_b) \cap L(v_c) = \emptyset$. 
    This also implies that $\lambda(v_av_b) \cap \lambda(v_cv_d) = \emptyset$. 
    Hence, the $\lambda$-sets of all edges in $M$ are pairwise disjoint. 
    As there are at most $\si(G)$ many non-empty pairwise disjoint label sets, we know that $|M| \leq \si(G)$ and, thus, $\lmim(G) \leq \si(G)$.    
\end{proof}


A \emph{tree decomposition} of a graph $G$ is a pair $(T,\{X_t\}_{t\in V(T)})$ consisting of a tree $T$ and a mapping asigning to each node $t\in V(T)$ a set $X_t\subseteq V(G)$ (called a \emph{bag}) such that the following conditions are satisfied: (i)~the union of all the bags equals $V(G)$, (ii)~for every edge $uv\in E(G)$ there exists a bag $X_t$ such that $u,v\in X_t$, and (iii)~for every vertex $v\in V(G)$ the bags containing $v$ form a subtree of $T$.
A \emph{path decomposition} of $G$ is a tree decomposition of $G$ such that $T$ is a path.
For simplicity, we will denote a path decomposition simply by the corresponding sequence $\cP = (X_1,\dots,X_k)$ of bags. 
Note also that in this case, condition (iii) simplifies to: for every vertex $v\in V(G)$ the bags containing $v$ form a consecutive subsequence of \cP.
The \emph{width} of a tree decomposition  is the maximal size of its bags minus 1. 
The \emph{treewidth} of a graph $G$, denoted by $\tw(G)$, is the minimum width of a tree decomposition of $G$.
The \emph{pathwidth}, denoted by $\pw(G)$, is defined analogously, with respect to path decompositions.
Yolov~\cite{zbMATH06850324} and independently Dallard et al.~\cite{zbMATH07796423} introduced the parameter called tree-independence number (or \emph{$\alpha$-treewidth}).
The \emph{independence number} of a tree decomposition $(T,\{X_t\}_{t\in V(T)})$ of a graph $G$ is defined as the maximum cardinality of an independent set $I$ in $G$ such that there exists a bag $X_t$ with $I\subseteq X_t$, or, equivalently, the maximum, over all bags $X_t$, of the independence number of the subgraph of $G$ induced by the bag $X_t$. 
The \emph{tree-independence number} of a graph $G$, denoted by $\tin(G)$, is defined as the minimum independence number of a tree decomposition of $G$.
We define the \emph{path-independence number}, denoted by $\pin(G)$, analogously, with respect to path decompositions.

We now prove a characterization of the path-independence number, which relies on the concept of the \emph{intersection} of two graphs $G_1 = (V_1,E_1)$ and $G_2 = (V_2,E_2)$, denoted by $G_1\cap G_2$ and defined as the graph $(V_1\cap V_2,E_1\cap E_2)$.
In the proof we will use the following two known facts about path decompositions and interval graphs (see~\cite{zbMATH01238685,zbMATH03214396}):
\begin{itemize}
\item Let $G$ be a graph, let $\cP$ be a path decomposition of $G$, and let $S\subseteq V(G)$ be a set of vertices of $G$ such that for every two vertices $u,v\in S$ there exists a bag $X_i$ of $P$ such that $u,v\in X_i$. 
Then there exists a bag $X_j$ of $\cP$ such that $S\subseteq X_j$.
\item A graph $G$ is an interval graph if and only if it admits a path decomposition in which each bag is a clique in $G$.
\end{itemize}

\begin{theorem}\label{tree-alpha-characterization}
Let $G$ be a graph.
Then, the path-independence number of $G$ equals the minimum integer $k\ge 0$ such that $G$ is the intersection of an interval graph and a graph with independence number at most $k$. 
\end{theorem}

\begin{proof}
Let use denote by $k$ the minimum nonnegative integer such that $G$ is the intersection of an interval graph and a graph with independence number at most~$k$.

We first show that $\pin(G)\le k$.
Let $G$ be the intersection of an interval graph $G_1$ and a graph $G_2$ with independence number at most $k$.
Each of the properties of being an interval graph and having independence number at most $k$ is preserved under vertex deletion, hence we may assume that $V(G_1) = V(G_2) = V(G)$.
Since $G_1$ is an interval graph, it admits a path decomposition $\cP$ such that each bag is a clique in $G_1$.
Note that $G$ is a spanning subgraph of $G_1$ and therefore $\cP$ is also a path decomposition of $G$.
To show that $\pin(G)\le k$, it suffices to show that for every bag, the subgraph of $G$ induced by the bag has independence number at most $k$.
Consider an arbitrary bag $X_i$ of $\cP$ and the subgraph of $G$ induced by the bag.
Note that this graph is the intersection of the subgraphs of $G_1$ and $G_2$ induced by $X_i$.
Since each bag is a clique in $G_1$, the graph $G_1[X_i]$ is complete and hence $G[X_i]$ equals the subgraph of $G_2$ induced by $X_i$.
This implies that $\alpha(G[X_i])\le \alpha(G_2)\le k$, which is what we wanted to show.

Next, we show that $k\le \pin(G)$.
Let $p = \pin(G)$ and let $\cP$ be a path decomposition of $G$ with independence number $p$.
We show that $G$ is the intersection of an interval graph and a graph with independence number at most~$p$; note that this will imply $k\le p$.
Let $G_1$ be the graph obtained from $G$ by adding edges so that each bag of the path decomposition $\cP$ becomes a clique and let $G_2$ be the graph obtained from $G$ by adding edges between any two vertices that are not contained in the same bag in $\cP$.
We claim that $G = G_1\cap G_2$, $G_1$ is an interval graph, and $G_2$ has independence number at most $p$.

First, observe that $\P$ is a path decomposition of $G_1$ in which each bag is a clique and, hence, by the first property above, $G_1$ is an interval graph.

To show that $\alpha(G_2)\le p$, consider an arbitrary independent set $I$ in $G_2$. 
Then $I$ is also an independent set in $G$ and, moreover, for any two vertices $u,v\in I$ there exists a bag $X_i$ of $\cP$ such that $u$ and $v$ both belong to that bag.
By the second property above, there exists a bag $X_j$ of $\cP$ such that $I\subseteq X_j$. 
Thus, since $\cP$ is a path decomposition of $G$ with independence number $p$, we infer that $|I|\le p$.
Since this holds for an arbitrary independent set of $G_2$, the independence number of $G_2$ is at most $p$.

It remains to verify that $G$ is the intersection of $G_1$ and $G_2$.
By construction, $G$ is a spanning subgraph of both $G_1$ and $G_2$.
Thus, $E(G)\subseteq E(G_1)\cap E(G_2)$.
Consider now two vertices $u$ and $v$ of $G$ that are adjacent in $G_1$.
Then, there exists a bag containing both $u$ and $v$. 
This implies that $u$ and $v$ are adjacent in $G_2$ if and only if they are adjacent in $G$. 
Hence, $G = G_1\cap G_2$.
\end{proof}

With a similar approach as that used to prove \cref{tree-alpha-characterization}, it can be proved that the tree-independence number of a graph $G$ equals the minimum integer $k\ge 0$ such that $G$ is the intersection of a chordal graph and a graph with independence number at most $k$.

\Cref{tree-alpha-characterization} has the following consequence.

\begin{corollary}\label{corol:si-pin}
Every graph $G$ satisfies $\pin(G) \leq \si(G)$.
\end{corollary}

\begin{proof}
Let $(R,L)$ be a $\si(G)$-simultaneous interval representation of $G$. Let $G_1$ be the interval graph with vertex set $V(G)$ and edge set $\{uv\colon u,v \in V(G), u\neq v,  R(u) \cap R(v) \neq \emptyset\}$. 
Furthermore, let $G_2$ be the graph with vertex set $V(G)$ in which two vertices $u$ and $v$ are adjacent if and only if $L(u) \cap L(v) \neq \emptyset$.
Then $G = G_1\cap G_2$, the graph $G_1$ is an interval graph and since any independent set in $G_2$ corresponds to a family of pairwise disjoint nonempty subsets of the label set $\{1,\ldots, \si(G)\}$, the graph $G_2$ has independence number at most $\si(G)$.
Therefore, $\pin(G)\le \si(G)$ by \cref{tree-alpha-characterization}.
\end{proof}



Note that complements of matchings have independence number 2 and, thus, also path-independence number at most 2. Due to \cref{lem:complement-matching-1}, they form a class of graphs with bounded path independence number but unbounded simultaneous interval number.

In the spirit of Dallard et al.~\cite{dallard2021treewidth}, we can also show that graphs with bounded simultaneous interval number are $(\pw, \omega)$-bounded (and consequently $(\tw, \omega)$-bounded; note that Chaplick et al.~\cite{chaplick2021topological} refer to the same property as the \emph{clique-treewidth property}). 
A graph class $\mathcal{G}$ is said to be \emph{$(\pw, \omega)$-bounded} (resp., \emph{$(\tw, \omega)$-bounded}) if there is a function $f$ such that for all graphs $G\in \mathcal{G}$ and all induced subgraphs $G'$ of $G$, it holds that $\pw(G') \leq f(\omega(G'))$ (resp., $\tw(G') \leq f(\omega(G'))$), where $\omega(G')$ is the clique number of $G'$.

\begin{theorem}\label{thm:binding}
    Every graph $G$ satisfies $\pw(G) \leq \si(G) \omega(G)-1$.
\end{theorem}

\begin{proof}
  Fix an $\si(G)$-simultaneous interval representation $(R,L)$ of $G$ and let $H$ be the interval graph corresponding to the interval representation $R$.
  It is clear that $V(G)=V(H)$ and $E(G) \subseteq E(H)$. The pathwidth of $H$ is then exactly $\omega(H)-1$ and $H$ admits a path decomposition $\mathcal{P}$ in which each of the bags is a clique, as $H$ is an interval graph. 
    Furthermore, the path decomposition $\mathcal{P}$ of $H$ is also a path decomposition of $G$, as $E(G) \subseteq E(H)$.

    On the other hand, we know that for each clique of $H$, say for example $C$, the intervals belonging to that clique have a common intersection in the interval representation of $G$. Consequently, any intervals of $C$ that share a common label in the interval model of $G$ must also form a clique in $G$. Therefore, any bag of $\mathcal{P}$ can have size at most $\si(G)\omega(G)$. 
    As $\mathcal{P}$ is also a path decomposition of $G$, we get $\pw(G) \leq \si(G) \omega(G)-1$.
\end{proof}

\subsection{Upper Bounds}

We begin our discussion on upper bounds by proving that bounded pathwidth implies bounded simultaneous interval number.

\begin{theorem}\label{thm:pw}
    Every graph $G$ satisfies $\si(G) \leq \pw(G)^2 + \pw(G)$.
\end{theorem}

\begin{proof}
\begin{figure}
\centering
    \begin{tikzpicture}[yscale=0.8]
        \tikzmath{\eps = 0.10; \i=1; \j=0;}
        \node (i) at (-\eps,\i) {$i$};
        \node (j) at (-\eps,\j) {$j$};

        \draw[very thick] (0.5+0.5*\eps,\i) -- (0.5+0.5*\eps,\j);
        \draw[very thick] (2.5+0.5*\eps,\i) -- (2.5+0.5*\eps,\j);
        \draw[very thick] (5+0.5*\eps,\i) -- (5+0.5*\eps,\j);
        \draw[very thick] (7.5+0.5*\eps,\i) -- (7.5+0.5*\eps,\j);
        \draw[very thick] (9.5+0.5*\eps,\i) -- (9.5+0.5*\eps,\j);
        \draw[very thick] (12.5+0.5*\eps,\i) -- (12.5+0.5*\eps,\j);
        
        \draw[line width=4, new-blue] (0 + \eps,\j) -- node[below, black] {$a_{ij}$} (4,\j);
        \draw[very thick, new-blue] (4 + \eps ,\j) -- node[below, black] {$b_{ij}$} (6,\j);
        \draw[very thick, new-blue] (6 + \eps, \j) -- node[below, black] {$\varnothing$} (7,\j);
        \draw[line width=4, new-blue] (7 + \eps, \j) -- node[below, black] {$b_{ij}$} (11,\j);
        \draw[very thick, new-blue] (11 + \eps, \j) -- node[below, black] {$\varnothing$} (12,\j);
        \draw[very thick, new-blue] (12 + \eps, \j) -- node[below, black] {$a_{ij}$} (13,\j);

        \draw[very thick, lipicsYellow] (0 + \eps,\i) -- node[above, black] {$a_{ij}$} (1,\i);
        \draw[very thick, lipicsYellow] (1 + \eps ,\i) -- node[above, black] {$\varnothing$} (2,\i);
        \draw[very thick, lipicsYellow] (2 + \eps, \i) -- node[above, black] {$a_{ij}$} (3,\i);
        \draw[line width=4, lipicsYellow] (3 + \eps, \i) -- node[above, black] {$b_{ij}$} (8,\i);
        \draw[very thick, lipicsYellow] (8 + \eps, \i) -- node[above, black] {$\varnothing$} (9,\i);
        \draw[very thick, lipicsYellow] (9 + \eps, \i) -- node[above, black] {$b_{ij}$} (10,\i);
        \draw[line width=4, lipicsYellow] (10 + \eps, \i) -- node[above, black] {$a_{ij}$} (13,\i);

    \end{tikzpicture}
    \caption{Illustration of the proof of \cref{thm:pw}. 
    The thick intervals mark the active intervals. 
    A black edge between two intervals means that the corresponding vertices are adjacent in $G$. 
    The symbol $\varnothing$ means that the corresponding vertex has neither label $a_{ij}$ nor the label $b_{ij}$. 
    However, the vertex will have other labels.}\label{fig:pw}
\end{figure}
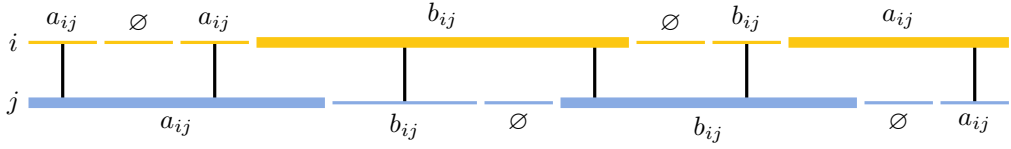
Let $k := \pw(G) + 1$. 
Consider a path decomposition $\cP$ of $G$ with maximal bag size $k$. 
It is easy to see that we can transform $\cP$ in such a way that every bag has size $k$. Furthermore, we can ensure that two consecutive bags differ only in two vertices, i.e., both vertices are part of exactly one of the two bags and all the other vertices are part of both bags or of none of them. 
This can be done by adding a sequence of new bags between two old ones in which the vertices are removed and introduced one by one.
Let $\cP' = (X_1, \dots, X_p)$ be the resulting path decomposition of $G$. 
Now there exists a mapping $f : V(G) \to \{1,\dots,k\}$ such that every bag of $\cP'$ contains a vertex $v$ with $f(v) = i$ for all $i \in \{1,\dots,k\}$. 
For every vertex of $G$, we define the interval $R(v)$ as $(a - \varepsilon, b + \varepsilon)$ where $0 < \varepsilon < \frac12$, $a$ is the smallest index such that $X_a$ contains $v$ and $b$ is the largest index such that $X_b$ contains $v$. 
It follows that the intervals of two vertices have a non-empty intersection if and only if these vertices are part of a common bag. 
Therefore, intervals of vertices with the same $f$-value have an empty intersection.
    
    It remains to show that we can label the vertices with at most $k \cdot (k-1)$ labels in such a way that the defined intervals form a simultaneous interval representation of $G$. 
    For every set $\{i,j\} \subseteq \{1,\dots,k\}$ with $i\neq j$, we introduce labels $a_{ij}$ and $b_{ij}$. Note that $a_{ij} = a_{ji}$ and $b_{ij} = b_{ji}$. 
    In the following we describe a procedure how to label the vertices of $G$ (see \cref{fig:pw} for an illustration). 
    During that labeling procedure, we will always have one \emph{active vertex} $\hat{v}$ and one \emph{active label} $c_{ij} \in \{a_{ij}, b_{ij}\}$. 
    To define the first active vertex let $x$ be the vertex with $f(x) = i$ whose interval ends first and let $y$ be the vertex with $f(y) = j$ whose interval ends first. 
    Without loss of generality, we may assume that $r(x) < r(y)$. 
    We define the first active vertex $\hat{v}$ to be $y$. The first active label $c_{ij}$ is $a_{ij}$. The active vertex $\hat{v}$ gets the label $c_{ij}$. For all vertices $z$ with $f(z) \in \{i,j\} \setminus f(\hat{v})$ and $\ell(\hat{v}) < r(z) < r(\hat{v})$, we add $c_{ij}$ to $L(z)$ if and only if $\hat{v}z \in E(G)$. Now consider the vertex $w$ with $f(w) \in \{i,j\} \setminus f(\hat{v})$ and $\ell(w) < r(\hat{v}) < r(w)$. Vertex $w$ becomes the new active interval. If $\hat{v}w \in E(G)$, then the active label stays the same, otherwise the new active label becomes the other one. In any case $w$ gets the new active label. Note that $L(\hat{v}) \cap L(w) \neq \emptyset$ if and only if $\hat{v}w \in E(G)$. 
    We repeat this procedure until the end of the interval representation. 
    Furthermore, we repeat the whole procedure for all sets $\{i,j\} \subseteq \{1,\dots, k\}$.
    In the end, we obtain a $d$-simultaneous interval representation $(R,L)$ of $G$ where $d= 2{k\choose 2} = k(k-1) = \pw(G)^2 + \pw(G)$.
\end{proof}

Observe at this point that bounded simultaneous interval number does not imply bounded pathwidth as is proven by the class of interval graphs.

An \emph{edge clique cover} of a graph $G$ is a set $\K$ of cliques of $G$ such that every edge of $G$ is contained in some clique of $\K$. 
We denote by $\eccn(G)$ the \emph{edge clique cover number} of $G$, that is, the minimum size of an edge clique cover of $G$.

\begin{lemma}\label{lemma:ecc-general}
    Let $\C$ be a class of intersection graphs. 
    Let $G$ be a graph, let $d\ge 0$ be an integer, and let $R$ be a $\C$-representation of some graph $F \in \C$. 
    Then, there exists a $d$-simultaneous $\C$-representation $(R, L)$ of $G$ if and only if there exists a graph $H$ with $\eccn(H) \leq d$ and $G$ is the intersection of $F$ and $H$.
\end{lemma}

\begin{proof}
First assume that $H$ is a graph with $\eccn(H) \leq d$ and $G = F \cap H$. 
Let $\K = \{C_1, \dots, C_d\}$ be an edge clique cover of $H$. 
For each vertex $v$ in $G$, we define $L(v)$ as follows: $L(v) := \{i\colon v \in C_i\}$. 
We claim that $(R,L)$ forms a $d$-simultaneous $\C$-representation of $G$. 
If two vertices $u,v \in V(G)$ are adjacent in $G$, then they are adjacent in $F$ and, thus, the sets $R(u)$ and $R(v)$ have a non-empty intersection. 
Furthermore, $u$ and $v$ are adjacent in $H$ and, hence, they are part of a common clique in $\K$. 
Therefore, the sets $L(u)$ and $L(v)$ share at least one label. 
If $u$ and $v$ are not adjacent in $G$, then either they are not adjacent in $F$ -- which implies that the sets $R(u)$ and $R(v)$ do not intersect -- or they are not adjacent in $H$ -- which implies that no clique contains both of them, and, hence their label sets are disjoint. 
Therefore, $(R,L)$ is a $d$-simultaneous $\C$-representation of $G$.

Now let $(R, L)$ be a $d$-simultaneous $\C$-representation of $G$. 
We define $H$ as the intersection graph of the label sets of $G$, i.e., $V(H) = V(G)$ and there is an edge between two distinct vertices $u$ and $v$ in $H$ if and only if $L(u) \cap L(v) \neq \emptyset$. 
It is easy to see that the intersection of $F$ and $H$ is the graph $G$. It remains to show that $\eccn(H) \leq d$. For any $i \in \{1,\dots, d\}$, let $C_i$ be the set of vertices of $G$ whose label set contains $i$. By definition of $H$, all the sets $C_i$ form cliques in $H$. Furthermore, the vertices of any edge in $H$ have to share some label and, thus, they are contained in a common set $C_i$. Therefore, the sets $C_i$ form an edge clique cover of $H$ of size $d$ implying $\eccn(H) \leq d$.
\end{proof}

\Cref{lemma:ecc-general} implies the following.

\begin{corollary}\label{cor:ecc-general}
Let $\C$ be a class of intersection graphs, let $G$ be a graph, and let $d\ge 0$ be an integer.
Then, $G$ has an \hbox{$d$-simultaneous} $\C$-representation if and only if $G$ is the intersection of a graph in $\C$ and a graph with edge clique cover number at most $d$.
\end{corollary}

\Cref{lemma:ecc-general} also implies the following strengthening of \Cref{thm:label-edges}.

\begin{sloppypar}
\begin{theorem}\label{thm:label-ecc}
For every class of intersection graphs $\C$, every graph $G$ has an \hbox{$\eccn(G)$-simultaneous} $\C$-representation.
\end{theorem}
\end{sloppypar}

\begin{proof}
By our assumption on classes of intersection graphs, there exists a set family $\mathcal{F} \in S_\C$ that contains some non-empty set $S$. 
Let $F$ be the complete graph with $|V(G)|$ vertices and let $H = G$.
Note that mapping every vertex of $F$ to $S$ yields a $\C$-representation of $F$, and, thus, $F$ belongs to $\C$.
Moreover, $G$ is the intersection of $F$ and $H$ and $\eccn(H) = \eccn(G)$.
Hence,  by \Cref{lemma:ecc-general}, $G$ admits an $\eccn(G)$-simultaneous $\C$-representation.
\end{proof}

\begin{corollary}\label{thm:si-eccn}
     Every graph $G$ satisfies $\si(G) \leq \eccn(G)$.
\end{corollary}

Interval graphs, and in particular paths, have unbounded edge clique cover number. Thus, bounded simultaneous interval number does not imply bounded edge clique cover number.

The bound given by \Cref{thm:si-eccn} is tight. 
Let us denote by $K_n^p$ the complete multipartite graph on $p$ partite sets of the same size $n$ and by $\lambda(n)$ the largest size of a family of mutually orthogonal Latin squares of order $n$. 
It is known that $\lambda(n)\le n-1$ and that equality holds if and only if there exists a projective plane of order $n$. 
Thus $\lambda(q) =q-1$ if $q$ is a prime power, but in general the exact computation of the value of $\lambda(n)$ is difficult.
Park, Kim, and Sano showed in~\cite{MR2558610} that for any two integers $p$ and $n$ such that $3\le p\le \lambda(n)+2$, the edge clique cover number of $K_n^p$ equals $n^2$.
Taking $p = 3$, we thus obtain, by combining \cref{thm:3-partite,thm:si-eccn}, that for the complete $3$-partite graph $G$ with parts of equal size, we have $\si(G) = \eccn(G) = \frac{|V(G)|^2}{9}$.

%% file: recognition.tex
\section{Complexity of Computing the Simultaneous Interval Number}\label{sec:si-complexity}

Using the characterization from \Cref{def:labeling}, we can state three natural recognition problems for $d$-simultaneous \C-representations. 

\begin{problem}[Simultaneous $\C$-Representation Problem]~
\begin{description}
\item \textbf{Input:} A graph $G$ and a labeling $L:V(G)\to {\mathcal P}(\{1,\ldots, d\})$ of $G$.
\item \textbf{Question:} Does there exist a $d$-simultaneous $\C$-representation $(R, L)$ of $G$?
\end{description}    
\end{problem}

By Observations~\ref{obs:labeling}  and~\ref{obs:labeling-2}, Problem~1 is a generalization of the simultaneous representation problems by Jampani and Lubiw~\cite{jampani2009simconf}.

In the second problem we are given the graph and some representation and want to find a suitable labeling.


\begin{problem}[Simultaneous Labeling Problem Given a $\C$-Representation]~
\begin{description}
\item \textbf{Input:} A graph $G$ and a $\C$-representation $R$ of a graph $F$ with $V(G) = V(F)$ and $E(G)\subseteq E(F)$.
\item \textbf{Question:} What is the smallest number $d \in \N$ such that there exists a $d$-simultaneous $\C$-representation $(R,L)$ of $G$?
\end{description}    
\end{problem}

In the third version, we are given just a graph and wish to compute the smallest number of labels needed for the graph to have a $d$-simultaneous $\C$-representation.


\begin{problem}[Generalized Simultaneous $\C$-Representation Problem]~
\begin{description}
\item \textbf{Input:} A graph $G$.
\item \textbf{Question:} What is the smallest number $d \in \N$ such that there exists a $d$-simultaneous $\C$-representation of $G$?
\end{description}    
\end{problem}



Recall the definition of a class of intersection graphs given on p.~\pageref{def-intersection}.

\begin{theorem}\label{thm:labeling-hard}
    The Simultaneous Labeling Problem Given a $\C$-Representation is \NP-hard for any class of intersection graphs $\C$, even if all sets in the given $\C$-representation pairwise intersect.
\end{theorem}

\begin{proof}
    We use the fact that it is \NP-hard to compute the edge clique cover number $\eccn(G)$ of a given graph $G$~\cite{kou1978covering}. 
    Let $G$ be a graph and let $F$ be the complete graph with $V(F) = V(G)$.
    By our assumption on classes of intersection graphs, there exists a set family $\mathcal{F} \in S_\C$ that contains some non-empty set $S$. 
    Consider the $\C$-representation $R$ of $F$ mapping all vertices of $F$ to $S$.  
    Due to \cref{lemma:ecc-general}, there exists a $d$-simultaneous $\C$-representation $(R,L)$ of $G$ if and only if there exists a graph $H$ with $\eccn(H) \leq d$ and $F \cap H = G$. 
    As $F$ is complete, it must hold that $H = G$. Therefore, there exists a $d$-simultaneous $\C$-representation $(R,L)$ of $G$ if and only if $\eccn(G) \leq d$. 
    Hence, the Simultaneous Labeling Problem Given a $\C$-Representation is \NP-hard even for $\C$-representations like $R$.
\end{proof}

\begin{theorem}
   The Generalized Simultaneous $\C$-Representation Problem is \NP-hard for every class of intersection graphs that is a subclass of the class of $C_4$-free graphs and contains the class of complete split graphs. 
\end{theorem}

\begin{proof}
    Let $\C$ be a class fulfilling the properties given in the theorem. 
    Due to \cref{thm:labeling-hard}, the Simultaneous Labeling Problem Given a $\C$-Representation is \NP-hard for $\C$ even if all sets in the given $\C$-representation pairwise intersect.
    We will reduce that problem to the Generalized Simultaneous $\C$-Representation Problem.

    Let $G$ be an arbitrary graph. We add $|E(G)| + 1$ pairwise non-adjacent vertices to $G$ and make them all adjacent to all vertices of $G$. We call the resulting graph $H$. Let $S$ be the set of these added vertices. We claim that for any $d \leq |E(G)|$, there exists a $d$-simultaneous $\C$-representation of $G$ whose elements pairwise intersect if and only if there exists a $d$-simultaneous $\C$-representation of $H$. 

    First assume that we have a $d$-simultaneous representation $(R_G,L_G)$ of $G$ where all sets in the image of $R_G$ pairwise intersect.
    Let $H'$ be the graph obtained from $H$ by making the vertices of $G$ a clique. Obviously, $H'$ is a complete split graph. Since $\C$ contains the complete split graphs, there is a $\C$-representation $R_H$ of $H'$. We define a label function $L_H$ of $H$ as follows: For all vertices $v \in V(G)$, we set $L_H(v) := L_G(v)$. For all vertices $v \in S$, we choose $L_H(v)$ to be the \emph{universal} label set of $L_G$, that is, the union of all label sets of $L_G$. 
    We claim that $(R_H, L_H)$ is a $d$-simultaneous representation of $H$. 
    If two vertices $x, y \in V(H)$ are adjacent, then they are also adjacent in $H'$. Thus, their representations in $R_H$ intersect. Furthermore, if both $x$ and $y$ are in $G$, then their label sets in $L_G$ are non-disjoint and, thus, their label sets in $L_H$ are non-disjoint. 
    If one of $x$ and $y$ is in $S$, then the label sets of $x$ and $y$ in $L_H$ are non-disjoint since one of them is universal. 
    If $x$ and $y$ are not adjacent in $H$, then either both are in $G$ or both are in $S$. If both are in $G$, then their label sets in $L_G$ (and, thus, $L_H$) are disjoint since their representations in $R_G$ are non-disjoint. If both vertices are in $S$, then they are also not adjacent in $H'$ and, thus, their representations in $R_H$ are disjoint. Summing up, $(R_H,L_H)$ is an $d$-simultaneous $\C$-representation of $G$. 

    Now assume that there exists a $d$-simultaneous $\C$-representation $(R_H, L_H)$ of $H$. 
    Since no vertex in $S$ is isolated in $H$, the label sets $L_H(v)$, $v\in S$, are all nonempty.
    There are at most $d$ pairwise disjoint nonempty label sets in $L_H$.
    Hence, since $S$ contains $|E(G)| + 1 > d$ pairwise non-adjacent vertices, there are at least two vertices $x$ and $y$ in $S$ such that $R_H(x)$ and $R_H(y)$ are disjoint.
    We claim that all the sets $R_H(v)$, $v\in V(G)$, have to pairwise intersect. 
    Assume for contradiction that this is not the case and let $u$ and $v$ be two vertices in $G$ such that $R_H(u)$ and $R_H(v)$ are disjoint. 
    As $u$ and $v$ are adjacent to both $x$ and $y$ in $H$, each of the sets $R_H(u)$ and $R_H(v)$ intersects each of the sets $R_H(x)$ and $R_H(y)$.
    This implies that the graph $F$ such that $R_H$ is a $\C$-representation of $F$ contains the induced cycle $(x,u,y,v)$; a contradiction to the fact that $F\in \C$ and $\C$ is a subclass of the class of $C_4$-free graphs. 
    Hence, the representations of all vertices of $G$ in $R_H$ pairwise intersect. 
    Let $(R_G, L_G)$ be the restriction of $(R_H, L_H)$ to the vertices of $G$. 
    Then, $(R_G, L_G)$ is a $d$-simultaneous $\C$-representation of $G$ where all the elements of $R_G$ pairwise intersect.
\end{proof}

\begin{corollary}
Let $\C$ be the class of interval graphs or the class of chordal graphs.
Then, the Generalized Simultaneous $\C$-Representation Problem is \NP-hard.
\end{corollary}

\begin{corollary}
    It is \NP-hard to compute the simultaneous interval number of a graph $G$.
\end{corollary}

%% file: clique.tex
\section{Cliques}\label{sec:cliques}

In this section, we focus on the \emph{Maximum Clique} problem: Given a graph $G = (V,E)$, compute a largest clique in $G$.
The problem can be naturally generalized to the weighted case, where the input graph is equipped with a vertex weight function $w:V\to \mathbb{Q}_+$ and the task is to find a clique $C$ in $G$ maximizing its weight, $w(C)$, defined as the sum of the weights of the vertices in $C$.

\begin{theorem}\label{thm:cliques}
A graph $G$ has at most $2^{2^{\si(G)}} \cdot n$ many maximal cliques.
\end{theorem}

\begin{proof}
Let $d = \si(G)$ and fix a $d$-simultaneous interval representation $(R,L)$ of $G$. 
Let $C$ be a maximal clique of $G$. 
There exists a point $p$ on the real line that is contained in any interval of the vertices contained in $C$. 
Furthermore, for every subset $S \subseteq \{1,\ldots,d\}$, if there is any vertex $u\in C$ such that $L(u) = S$, then the clique $C$ contains all the vertices $v$ whose label set is exactly $S$ and whose interval $R(v)$  contains $p$.
    There are at most $n$ points on the real line such that the sets of intervals containing these points are pairwise incomparable with respect to inclusion. 
    These are always points before the endpoint of some interval.
    For each of those points we have to decide for every subset of $\{1,\dots,d\}$ if vertices having this subset as label set are contained in the maximal cliques. There are $2^d$ many subsets. Therefore, there are $2^{2^d}$ different decisions and, thus, there are at most $2^{2^d}n$ many maximal cliques.
\end{proof}

\begin{theorem}\label{thm:maximal-cliques}
    Given a graph $G$ with $n$ vertices and a $d$-simultaneous interval representation of $G$, the maximal cliques of $G$ can be enumerated in time $\O(d\cdot 2^{2^d+2d} \cdot n \log n)$.
\end{theorem}

\begin{proof}
    We generate all binary vectors with $2^d$ many entries. Every entry stands for some label set. 
    We only keep those vectors where the label sets with entry 1 are pairwise non-disjoint.
    This can be checked in time $\O(d 2^{2d})$ per vector by comparing the label sets with entry 1 pairwise. 
    For every of the remaining binary vectors, we create the interval representation only containing the intervals whose vertices have a label set with entry 1. 
    This can be done in total time $\O(2^{2^d}\cdot n)$.
    Finally, we have to compute for every of those interval representations the set of maximal cliques which can be done in time $\O(n \log n)$~\cite{gupta1982efficient}.
\end{proof}

This result implies directly that we can compute a maximum-weight clique of a graph $G$ within the same time bound.

\begin{corollary}\label{thm:maximum-clique}
    Given a vertex-weighted graph $G$ with $n$ vertices and a $d$-simultaneous interval representation of $G$, we can find a maximum weight clique of $G$ in time $\O(d \cdot 2^{2^d+2d} \cdot n \log n)$.
\end{corollary}

Tsukiyama et al.~\cite{tsukiyama1977new} gave an algorithm that generates all maximal cliques in time $\O(n^3\mu)$ where $\mu$ is the number of maximal cliques. 
Using this algorithm, we can drop the requirement in \cref{thm:maximal-cliques,thm:maximum-clique} that the input graph is given together with a $d$-simultaneous interval representation.

\begin{theorem}
Given a vertex-weighted graph $G$ with $n$ vertices, we can find a list of all maximal cliques and a maximum weight clique of $G$ in time $\O(2^{2^{\si(G)}} \cdot n^3)$.
\end{theorem}

Let us remark that a faster dependency on $n$ (although still slower than quadratic in $n$) could be obtained by using some of the more recent maximal clique enumeration algorithms (see, e.g.,~\cite{MR3864713}). 

Note that the unweighted maximum clique problem is already \NP-hard for 2-unit interval graphs and 3-track interval graphs~\cite{francis2015maximum} while it is polynomial-time solvable for 2-track interval graphs. 
However, there is an \FPT{} algorithm for the clique problem on $d$-interval graphs when parameterized by $d$ plus solution size~\cite{fellows2009parameterized}.
 
Next we prove that the bound given in \cref{thm:cliques} is tight. To this end, we consider complements of matchings. As we have seen in \cref{lem:complement-matching-1}, the simultaneous interval number of those graphs is at least $\log_2(n-1)$ where $n$ is the number of vertices. Here, we show that this bound is tight. Let $M_m$ be the complement of a matching with $m$ edges. Gregory and Pullman~\cite{MR0676866} showed that $\lim_{m\to \infty}\frac{\eccn(M_m)}{\log_2(m)} = 1$. As we have seen in \cref{thm:si-eccn}, it holds that $\si(G) \leq \eccn(G)$. This implies the following result.

\begin{lemma}\label{lemma:matching-labels}
    For any $\varepsilon>0$, there exists some $n' \in \mathbb{N}$ such that for all even $n \geq n'$, the following holds: If $G$ is the complement of a matching with $n$ vertices, then $\si(G) \leq (1+\varepsilon)\log_2 n$.
\end{lemma}



Using this result, we are able to prove that the bound given in \cref{thm:cliques} is tight.

\begin{theorem}\label{thm:comatch}
     For any $\varepsilon$ with $0 < \varepsilon < 1$ and any $k \in \N$, there is an infinite family $\cal F$ of graphs such that any graph $G \in \cal F$ with $n$ vertices has at least $2^{2^{(1-\varepsilon)\si(G)}} \cdot n^k$ many maximal cliques.
\end{theorem}

\begin{proof}
    Let $n'$ be chosen as in  \cref{lemma:matching-labels}. Since $n^k \in o(2^{n^{1-\varepsilon^2}})$ and $n^{1-\varepsilon^2} \in o(n)$, there is an integer $n'' \geq n'$ such that for all integers $n \geq n''$ it holds that $n^k \leq 2^{n^{1-\varepsilon^2}}$ and $2^{2n^{1-\varepsilon^2}} \leq 2^{\frac n2}$.
    Let $n \geq n''$ be an even number and let $G$ be a complement of a matching with $n$ vertices.
    Then it holds:
    \begin{align*}
        2^{2^{(1-\varepsilon)\si(G)}} \cdot n^k &\leq 2^{2^{(1-\varepsilon)(1+\varepsilon)\log_2(n)}} \cdot n^k\\
                                                &= 2^{n^{1-\varepsilon^2}} \cdot n^k \\
                                                &\leq 2^{n^{1-\varepsilon^2}} \cdot 2^{n^{1-\varepsilon^2}} \\
                                                &= 2^{2n^{1-\varepsilon^2}} \\
                                                &\leq 2^{\frac n2}
    \end{align*}
     As $2^{\frac n2}$ is the number of maximal cliques of $G$, this proves the theorem.
\end{proof}

This result shows that the bound on the running time of our approach for the Maximum Clique problem cannot be significantly improved. 
Furthermore, the following result shows that the Maximum Clique problem cannot be solved with a single-exponential \FPT{} algorithm parameterized by the simultaneous interval number.

\begin{theorem}
    Unless $\P = \NP$, for any fixed $k \in \N$ there is no algorithm that solves the Maximum Clique problem on complements of cubic graphs with $n$ vertices in time $2^{\O(\si)} n^k$. 
\end{theorem}

\begin{proof}
    Alon~\cite{alon1986covering} showed that complements of cubic graphs have edge clique cover number of at most $c \log n$ for some constant $c$. 
    Due to \cref{thm:si-eccn}, their simultaneous interval number is also at most $c \log n$.
    Suppose for a contradiction that for some $k\in \N$, there is a $2^{\O(\si(G))} n^k$ algorithm for the Maximum Clique problem on complements of cubic graphs.
    Then we could solve the Maximum Clique problem on the complements of cubic graphs in polynomial time.
    However, as Mohar~\cite{mohar2001face} showed, the Maximum Independent Set problem is \NP-complete on cubic graphs, implying that the Maximum Clique problem is \NP-hard on their complements.
\end{proof}

Note that the above result does not rule out the possibility that it may be possible to solve the Maximum Clique problem in time $2^{\O(d)}n^k$ when a $d$-simultaneous interval representation of the graph is given. 

As graphs with bounded simultaneous interval number are $(\pw, \omega)$-bounded (\cref{thm:binding}), we can use the results from Chaplick et al.~\cite[Theorem~11]{chaplick2021topological} to show that the clique problem admits an \FPT{} algorithm when parameterized by the simultaneous interval number plus solution size.

\begin{theorem}
Given an $n$-vertex graph $G$ and an integer $k$, it can be determined in time $2^{\O(\si(G)k)} n$ whether $G$ contains a clique of size $k$.
\end{theorem}

%% file: coloring.tex
\section{Coloring}\label{sec:coloring}

Circular-arc graphs have linear mim-width at most~$2$~\cite[Lemma 4]{belmonte2013graph}, path-independence number at most~2~\cite[proof of Theorem 4.5]{milanic2022tree} and track number at most~2. 
Since the Coloring problem is \NP-hard on circular-arc graphs~\cite{garey1980complexity}, the same holds for graphs whose linear mim-width, path-independence number, and track number are at most 2. 
This result does not transfer directly to the simultaneous interval number, as the simultaneous interval number of complements of matchings and, thus, of circular-arc graphs is unbounded, due to \cref{lem:complement-matching-1}. 
Nevertheless, we can adapt a proof for the \NP-hardness of the Coloring problem on circular-arc graphs given by Marx~\cite{marx2003short} to the case of graphs of simultaneous interval number 2. 
This proof uses the following definitions and results.

\begin{problem}[Disjoint Paths]~
\begin{description}
\item \textbf{Input:} Directed graphs $G$ and $H$ on the same vertex set.
\item \textbf{Question:} Are there paths $P_e$ in $G$ for each $e \in E(H)$ such that these paths
are edge disjoint and path $P_e$ together with edge $e$ forms a directed cycle?
\end{description}    
\end{problem}

Given a directed graph $G = (V,E)$, the \emph{in-degree} (resp.~\emph{out-degree}) of a vertex $v\in V$ in $G$ is the number of directed edges $(x,y)\in E$ such that $v = y$ (resp.~$v = x$), and the \emph{degree} $d_G(v)$ of $v$ in $G$ is the number of directed edges $(x,y)\in E$ such that $v \in \{x,y\}$.
A directed graph $G = (V,E)$ is \emph{Eulerian} if for each vertex $v\in V$, the in-degree of $v$ equals its out-degree.

\begin{theorem}[Vygen \cite{vygen1995np}]\label{thm:disjoint-paths-np}
    The Disjoint Paths problem remains \NP-complete even if $G$ is acyclic and $G + H$ is Eulerian.
\end{theorem}

\begin{lemma}[Marx~\cite{marx2003short}]\label{lemma:disjoint-paths-cover}
    If $G + H$ is Eulerian and $G$ is acyclic, then every solution of the Disjoint Path problem given $G$ and $H$ uses every edge of $G$.
\end{lemma}

\begin{lemma}
    The Disjoint Paths problem remains \NP-complete even if $G$ is acyclic, $G + H$ is Eulerian, and every vertex in $H$ has degree at most one.
\end{lemma}

\begin{proof}
    We show that for every instance of the disjoint path problem there is an equivalent instance where every vertex is incident to at most one edge in $H$. 
    Let $\xi(H) := \sum_{v \in V} \max\{0,d_H(v) - 1\}$. 
    We prove our claim by induction on $\xi(H)$.
    
    If $\xi(H) = 0$, then for every vertex $v \in V$ it holds that $d_H(v) \leq 1$ and we are done. Thus, we may assume that for all instances with $\xi(H) \leq i$ the claim holds true. Let $(G,H)$ be an instance with $\xi(H) = i+1$ and let $w$ be a vertex with $d_H(w) \geq 2$. 
    
    First assume that there is an edge $(v,w) \in E(H)$. Then we insert a vertex $w'$ to $G$ and $H$. Furthermore, we add the edge $(w',w)$ to $G$ and we replace $(v,w)$ in $H$ by $(v,w')$. We call the newly obtained graphs $G'$ and $H'$.
    The graph $G'$ is acyclic since $w'$ is only incident to one edge in $G'$ and, thus, cannot be part of a cycle. 
    Furthermore, the graph $G' + H'$ is Eulerian since $v$ has lost one outgoing edge and obtained a new outgoing edge, $w$ has lost one incoming edge and obtained a new incoming edge, and $w'$ has one outgoing edge and one incoming edge. 

    If there is a solution of the disjoint path problem for the instance $(G,H)$, then we can replace the path $P_{(v,w)}$ by the path $(w',w) + P_{(v,w)}$ and this path closes a cycle with edge $(v,w')$. 
    On the other hand, if $(G',H')$ has a solution, then the path $P_{(v,w')}$ contains the edge $(w',w)$ and removing this edge leads to a path that closes a cycle with edge $(v,w)$. 
    Therefore, the instance $(G',H')$ is equivalent to the instance $(G,H)$. Since $\xi(H') = i$, there is an instance $(G'', H'')$ with $\xi(H'') = 0$ that is equivalent to $(G',H')$ and, thus, to $(G,H)$.

    If there is no edge $(v,w) \in E(H)$, then there must be an edge $(w,v) \in E(H)$. 
    In that case we add a vertex $w'$ to both $G$ and $H$, add the edge $(w,w')$ to $G$ and replace the edge $(w,v)$ in $H$ by $(w',v)$. 
    With the same argumentation as above it follows that the claim holds true.

 Since the desired equivalent instance can be computed from $H$ in polynomial time, by performing at most $\xi(H)$ of the above modification steps, the lemma follows.
\end{proof}

\begin{theorem}\label{thm:coloring-np}
The Coloring problem is \NP-complete on graphs $G$ with $\si(G) \leq 2$ even if a $2$-simultaneous interval representation of $G$ is given.
\end{theorem}

\begin{proof}
    We adapt a proof given by Marx~\cite{marx2003short} to establish \NP-completeness of the Coloring problem on circular-arc graphs. 
    Let $(G,H)$ be an instance of the Disjoint Paths problem such that $G$ is acyclic, $G+ H$ is Eulerian, and $d_H(v) \leq 1$ for all $v \in V(G)$. Let $k = |E(H)|$.

    Let $v_1, \dots, v_n$ be a topological sort of $G$. For every edge $(v_i, v_j) \in E(G)$ we construct an interval $(i,j)$ with label set $\{1\}$. Note that $i < j$, due to the property of the topological sort. 
    For every edge $(v_i,v_j) \in E(H)$ we may assume that $i > j$ since otherwise there is no path from $v_j$ to $v_i$ in $G$. We add the intervals $(0,j)$ and $(i,n+1)$ with label set $\{1,2\}$ and the interval $(j,i)$ with label set $\{2\}$. We call the resulting $2$-simultaneous interval graph $G'$.

    We claim that $(G,H)$ is a yes instance of the disjoint path problem if and only if $G'$ can be colored with $k$ colors. 
    First assume that $(G,H)$ is a yes instance. 
    Fix a solution, that is, paths $P_e$ in $G$ for each $e \in E(H)$ such that these paths are edge-disjoint and path $P_e$ together with edge $e$ forms a directed cycle.
    By \cref{lemma:disjoint-paths-cover}, the solution covers all the edges with $k$ directed cycles. 
    Let $C$ be the $\ell$-th cycle in the solution. For every edge $(v_i, v_j) \in E(C) \cap E(G)$ we color the corresponding interval $(i,j)$ that has label $\{1\}$ with color $\ell$. 
    For the edge $(v_i, v_j) \in E(C) \cap E(H)$ we color with color $\ell$ the intervals $(0,j)$ and $(i,n+1)$ that have label set $\{1,2\}$  as well as the interval $(j,i)$ that has label set $\{2\}$. This leads to a proper coloring of $G'$ since the only intervals with the same color that intersect each other do not share a label and, thus, their corresponding vertices are not adjacent.

    Now assume the graph $G'$ can be properly colored with $k$ colors. 
    As all the $k$ intervals with label set $\{1,2\}$ that start in $0$ pairwise intersect, they have different colors. 
    Now consider the subgraph of $G$ induced by the intervals containing label~$2$. 
    Since every vertex in $H$ has degree at most one, whenever an interval ends before point $n+1$, there is no other interval that ends at this point. 
    Furthermore, there is exactly one interval that starts at this point. 
    This implies that every point $p$ in the interval $(0,n+1)$ in which no interval ends belongs to exactly $k$ intervals. 
    Consequently, any two intervals such that the second one starts where the first one ends must have the same color.
    This implies, in particular, that the two intervals with label set $\{1,2\}$ representing the same edge of $H$ have the same color. 
    
    Now consider all the intervals that contain the label~$1$. 
    There are $k$ of those intervals that start in point $0$. 
    If exactly $j$ of those intervals end in point $i$, then there are exactly $j$ intervals that start in $i$, due to the Eulerian property of $G + H$. 
    Thus, any non-integer point in $(0,n+1)$ is contained in exactly $k$ intervals. 
    This also implies that for any of those points there is an interval with color $d \in \{1,\ldots, k\}$. 
    Therefore, the intervals with color $d$ represent a directed cycle in $G + H$ containing exactly one edge of $H$.
    Thus, $(G,H)$ is a yes instance of the disjoint path problem.
\end{proof}

As any class of graphs with bounded simultaneous interval number are $(\pw, \omega)$-bounded (\cref{thm:binding}), we can use the results from Chaplick et al.~\cite[Theorem~12]{chaplick2021topological} to show that the List $k$-Coloring problem admits an \FPT{} algorithm when parameterized by $k$ plus the simultaneous interval number.

\begin{theorem}
    Given a graph $G$, we can solve the List $k$-Coloring problem on $G$ in time $k^{\O(\si(G)k)} n$.
\end{theorem}

%% file: domination.tex
\section{Domination and Independent Sets}\label{sec:domination-independence}

The Dominating Set problem and the Independent Set problem can be solved in polynomial time on interval graphs~\cite{fanica1972algorithms,ramalingam1988unified}. 
However, when we parameterize these problems by the solution size and linear mim-width they are \W[1]-hard~\cite{jaffke2019mimiii}. If we parameterize the Dominating Set problem by $\tin$ and the solution size then it is \W[2]-hard~\cite{liu2011parameterized}. In contrast, when the problems are parameterized by simultaneous interval number and the solution size, then bounded-search-tree methods lead to FPT-algorithms.

\begin{theorem}\label{thm:fpt-dom}
    Given a graph $G$ with $n$ vertices and a $d$-simultaneous interval representation of $G$, we can decide whether $G$ has a dominating set of size at most $k$ or an independent set of size $k$ in time $\O(2^{kd} \cdot n)$.
\end{theorem}

\begin{proof}
    We use the common technique of bounded search trees (see, e.g.,~\cite{cygan2015param}).
    First consider the Dominating Set problem. 
    If the graph has no vertices, the algorithm returns ``yes'' .
    Otherwise, if $k = 0$, the algorithm returns ``no''.
    For $k\ge 1$ and $V(G)\neq \emptyset$, our algorithm considers the vertex $v$ whose interval ends first. 
    For every possible label set $S$ we do the following:
    If there exists a vertex in $N_G[v]$ with label set $S$, then let $u$ be the vertex in $N_G[v]$ with label set $S$ whose interval ends last. 
    Remove $u$ and all its neighbors from $G$ and solve the Dominating Set problem with parameter $k-1$ on the remaining graph recursively. 

    To prove that the algorithm works correctly it is sufficient to show that in some of the branches the algorithm detects a dominating set of size at most $k$ if there is one.
    We prove this by induction on $k$ and, therefore, we assume that the algorithm works correctly for $k-1$. 
    Let $D$ be a dominating set of size at most $k$ in $G$. 
    We know that there is a vertex $w \in N_G[v] \cap D$. 
    Let $S$ be the label set of $w$. 
    Let $u$ be the vertex in $N_G[v]$ with label set $S$ whose interval ends last. 
    We claim that $(D \setminus \{w\}) \cup \{u\}$ is also a dominating set of $G$. 
    Assume for contradiction that this is not the case. 
    Then there is a vertex $x$ that belongs to $N_G[w]\setminus N_G[u]$.
    As $w$ and $u$ have the same label set and $u$ does not end before $w$, we know that $x$ has to end before $u$ starts. 
    However, this implies that $x$ ends before $v$; a contradiction to the choice of $v$.

    For the Independent Set problem, if $k = 0$, the algorithm returns ``yes''.
    Otherwise, if the graph has no vertices and $k \ge 1$, the algorithm returns ``no''.
    For $k\ge 1$ and $V(G)\neq \emptyset$, our algorithm considers for every label set $S$, the vertex $u$ with label set $S$ whose interval ends first. 
    The algorithm removes $u$ and all its neighbors from $G$ and solves the Independent Set problem with parameter $k-1$ on the remaining graph recursively.

    Again, we prove the correctness by induction. Let $I$ be an independent set of size $k$. 
    Let $w$ be the vertex in $I$ whose interval ends first and let $S$ be the label set of $w$. 
    Let $u$ be the vertex with label set $S$ whose interval ends first. 
    Again, we claim that $(I \setminus \{w\}) \cup \{u\}$ is an independent set. 
    Assume for contradiction that this is not the case. 
    Then there is a vertex $x \in I\setminus\{w\}$ that is adjacent to $u$ but not to $w$. 
    As $u$ and $w$ have the same label set and $u$ does not end after $w$, $x$ has to end before $w$; a contradiction to the choice of $w$.

    For the running time observe that every call of our algorithm produces at most $2^d$ new calls and in any of these calls the parameter $k$ is decreased by one. Therefore, we have at most $2^ {kd}$ calls and any of these calls needs $\O(n)$ steps to find vertex $u$ and to remove it together with all its neighbors.
\end{proof}

Using a technique due to Fomin et al.~\cite{fomin2020tractability}, in~\cite{jaffke2019mimiii}  Jaffke et al. showed that a whole range of domination-type problems (including dominating and independent set) are $\W[1]$-hard when parameterized by mim-width and solution size. While that approach cannot be easily adapted for the simultaneous interval number, it is possible to show that at least one of these problems is $\W[1]$-hard when parameterized just by $\si$.

\begin{problem}{Independent Dominating Set Problem (IDSP)}
\begin{description}
\item[\textbf{Instance:}] A graph $G$ and an integer $k$.
\item[\textbf{Question:}]
Does there exist a set $X$ of at most $k$ vertices of $G$ such that $G[X]$ is edgeless and $N_G[X] = V(G)$?
\end{description}
\end{problem}


The results in~\cite{fomin2020tractability} use a reduction from the \emph{Multicolored Clique problem} (MCP), a technique popularized by Fellows et al.~\cite{fellows2009parameterized}. 
We will use a reduction from the \emph{Multicolored Independent Set problem}.

\begin{problem}{Multicolored Independent Set Problem (MISP)}
\begin{description}
\item[\textbf{Instance:}] A graph $G$ with a proper coloring of $k$ colors.
\item[\textbf{Question:}]
Is there an independent set $I$ in $G$ such that $I$ contains exactly one vertex of each color?
 \end{description}
\end{problem}

The MCP (and thus the MISP) was shown to be $\W[1]$-hard when parameterized by solution size  by Pietrzak~\cite{pietrzak2003param} and by Fellows et al.~\cite{fellows2009parameterized}. 
In fact, in \cite{cygan2015param, lokshtanov2011lower} the authors show the following result under the assumption of the Exponential Time Hypothesis (ETH) which asserts that solving $n$-variable $3$-SAT requires time $2^{\Omega(n)}$ (see~\cite{MR1820597}).

\begin{theorem}[Cygan et al.~\cite{cygan2015param}, Lokshtanov et al.~\cite{lokshtanov2011lower}]\label{thm:lower}
Assuming the Exponential Time Hypothesis, there is no $f(k) n^{o(k)}$ time algorithm for the MCP  (MISP) for any computable function $f$.
\end{theorem}

For an instance $G$ of the MCP we can assume that all color classes are of the same size $q$, since adding isolated vertices does not affect the existence or nonexistence of a multicolored clique.
A similar assumption can be made for the MISP.
For each color class $i \in \{1, \ldots , k\}$, we denote the vertices in the class by $v^{i}_{1}, \ldots , v^{i}_{q}$. 

We will show that the IDSP is $\W[1]$-hard when parameterized by the simultaneous interval number. 
To this end we will construct a reduction from the MISP in the following way. 
Let $G$ together with a vertex partition $V(G) = V_1 \dot{\cup} \ldots \dot{\cup} V_k$ be an instance of the MISP, where $V_i = \{v_1^i, \ldots , v_q^i\}$ for all $i \in \{1, \ldots , k\}$. 

Let $E(G) = \{e_1, \ldots , e_m\}$. We will now define a $(k+2)$-simultaneous interval graph $G'$ together with its $(k+2)$-simultaneous interval representation. 
For each vertex $v_j^i \in V(G)$ we will define a collection of $m+1$ (open) intervals (see \cref{fig:construction}) \[W^i_j :=\left\{\left(\gamma - 1 + (j-1) \frac{1}{q} + i \epsilon,\gamma + (j-1) \frac{1}{q} + i \epsilon\right)\colon \gamma \in \{1, \ldots, m+1\}\right\},\] 
where $\epsilon \ll \frac{1}{kq}$, i.e., all $k$ shifts in sum are much smaller than one interval of a $W^i_j$.

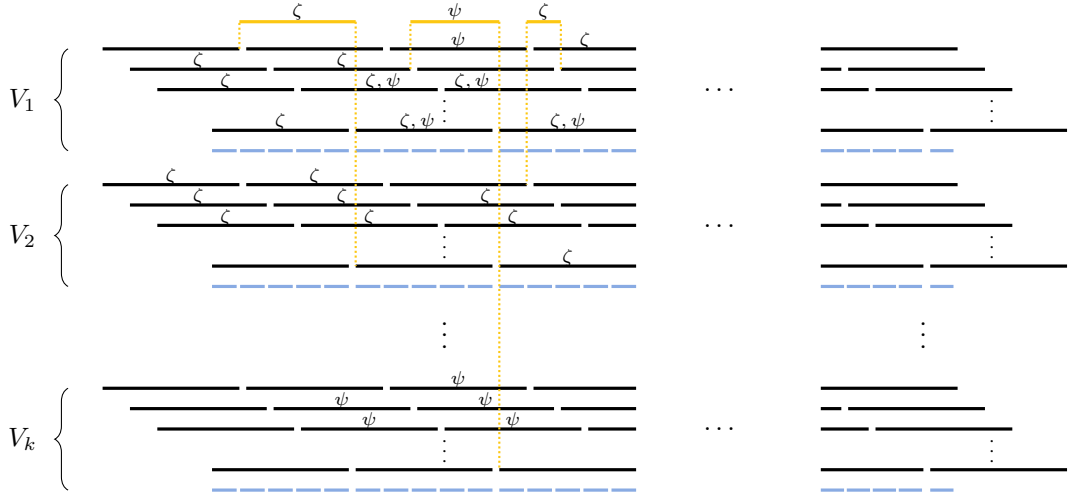
\begin{figure}
    \centering
    \begin{tikzpicture}[scale =0.9]
        \foreach \k in {0,2,5}
        {
            \foreach \i in {0,2.1,4.2,10.5}
            { 
                
                \foreach \j/\l in {0/0,0.3/0.4,0.6/0.8,1.2/1.6}  
                {
                    \draw[very thick](0+\i+\l,3-\j-\k) -- (2+\i+\l,3-\j-\k);
                }
                \newdimen\dummyDim
                \dummyDim = \i pt
                \ifdim \dummyDim < 10pt 
                \foreach \j in {0,...,4}
                {
                    \draw[very thick,new-blue](1.6+\i + \j *0.41,1.5-\k) -- (1.6+ \j*0.41+ 0.36 +\i,1.5-\k);

                }
                \fi
            }
            \foreach \l/\j in {0/0,0.4/0.3,0.8/0.6}
            {
                \draw[very thick] (6.3+ \l,3-\k-\j) -- (7.8, 3-\k -\j);
            }

            \foreach \l/\j in {0.4/0.3,0.8/0.6,1.6/1.2}
            {
                \draw[very thick] (10.5,3-\k-\j) -- (10.4 + \l, 3-\k -\j);
            }
            \foreach \j in {0,...,3}
            {
                \draw[very thick,new-blue](10.5 + \j *0.38,1.5-\k) -- (10.5+ \j*0.38+ 0.34,1.5-\k);
                
            }
            \draw[very thick,new-blue](12.1,1.5-\k) -- (12.44,1.5-\k);
            
            \path (13,2.3- \k) -- (13,1.8- \k) node [font= \scriptsize,midway,sloped] {$\dots$};
            \path (5,2.3- \k) -- (5,1.8- \k) node [font= \scriptsize,midway,sloped] {$\dots$};
            \path (6.8,2.4-\k) -- (11.3,2.4-\k) node [midway,sloped] {$\dots$}; 
        }

    \draw[thick,lipicsYellow,densely dotted] (2,3) -- (2,3.4);
    \draw[very thick,lipicsYellow] (2,3.4) -- (3.7,3.4);
    \path (2,3.55) -- (3.7,3.55) node [font= \scriptsize,midway,sloped] {$\zeta$};
    \draw[thick,lipicsYellow,densely dotted] (3.7,3.4) -- (3.7,-0.2);
    \draw[very thick,lipicsYellow] (4.5,3.4) -- (5.8,3.4);
    \path (4.5,3.55) -- (5.8,3.55) node [font= \scriptsize,midway,sloped] {$\psi$};
    \draw[thick,lipicsYellow, densely dotted] (4.5,2.7) -- (4.5,3.4);
    \draw[thick, lipicsYellow, densely dotted] (5.8,3.4) -- (5.8,-3.2);
    \draw[thick,lipicsYellow,densely dotted] (6.7,2.7)-- (6.7,3.4);
    \draw[thick,lipicsYellow,densely dotted] (6.2,3.4)-- (6.2,1);
    \draw[very thick,lipicsYellow] (6.7,3.4) -- (6.2,3.4);
    \path (6.7,3.55) -- (6.2,3.55) node [font= \scriptsize,midway,sloped] {$\zeta$};

    \path (4.2,3.125) -- (6.2,3.125) node [font= \scriptsize,midway,sloped] {$\psi$};
    
    \path (2.3,2.825) -- (4.7,2.825) node [font= \scriptsize,midway,sloped] {$\zeta$};
    
    \path (6.3,3.125) -- (7.8,3.125) node [font= \scriptsize,midway,sloped] {$\zeta$};
    
    \foreach \l\j in {0.4/0.3, 0.8/0.6, 1.6/1.2}
    {
    \path (0+\l,3.125-\j) -- (2+\l,3.125-\j) node [font= \scriptsize,midway,sloped] {$\zeta$};
    }

    \foreach \l\j in {1/0.6, 1.5/1.2}
    {
    \path (2.1+\l,3.125-\j) -- (4.1+\l,3.125-\j) node [font= \scriptsize,midway,sloped] {$\zeta,\psi$};
    }

    \foreach \l\j in {0.2/0.6, 1.6/1.2}
    {
    \path (4.2+\l,3.125-\j) -- (6.2+\l,3.125-\j) node [font= \scriptsize,midway,sloped] {$\zeta,\psi$};
    }


    \foreach \l\j in {0/0,0.4/0.3, 0.8/0.6}
    {
    \path (0+\l,1.125-\j) -- (2+\l,1.125-\j) node [font= \scriptsize,midway,sloped] {$\zeta$};
    }

    \foreach \l\j in {0/0,0.4/0.3, 0.8/0.6}
    {
    \path (2.1+\l,1.125-\j) -- (4.1+\l,1.125-\j) node [font= \scriptsize,midway,sloped] {$\zeta$};
    }

    \foreach \l\j in {0.4/0.3, 0.8/0.6, 1.6/1.2}
    {
    \path (4.2+\l,1.125-\j) -- (6.2+\l,1.125-\j) node [font= \scriptsize,midway,sloped] {$\zeta$};
    }

     \foreach \l\j in {0.4/0.3, 0.8/0.6}
    {
    \path (2.1+\l,-1.875-\j) -- (4.1+\l,-1.875-\j) node [font= \scriptsize,midway,sloped] {$\psi$};
    }
    
    \foreach \l\j in {0/0,0.4/0.3,0.8/0.6}
    {
    \path (4.2+\l,-1.875-\j) -- (6.2+\l,-1.875-\j) node [font= \scriptsize,midway,sloped] {$\psi$};
    }

    \path (5,-1) -- (5,-1.5) node [midway,sloped] {$\dots$};
    \path (12,-1) -- (12,-1.5) node [midway,sloped] {$\dots$};

    \draw [decorate,decoration={brace,amplitude=5pt,mirror,raise=3ex}] (0,3) -- (0,1.5) node[midway,xshift=-3em]{$V_1$};

    \draw [decorate,decoration={brace,amplitude=5pt,mirror,raise=3ex}] (0,1) -- (0,-0.5) node[midway,xshift=-3em]{$V_2$};

    \draw [decorate,decoration={brace,amplitude=5pt,mirror,raise=3ex}] (0,-2) -- (0,-3.5) node[midway,xshift=-3em]{$V_k$};
    
    \end{tikzpicture}
\caption{The yellow intervals represent the edges of $G$, the black intervals are the intervals of the $W^i_j$. 
The blue intervals are in $S_i$. 
Each of the rows marked $V_i$ represent that vertex set of $G$. For visual reasons the intervals belonging to the $V_i$ have not been shifted by $\epsilon$ as in the definition. 
For the same reason, we define $\zeta := k+1$ and $\psi:= k+2$. The labels of the edge intervals are denoted completely above these. 
Each of the other intervals also contains the label $i$ if it is associated with $V_i$. 
The intervals on the right have not been labeled.}
    \label{fig:construction}
\end{figure}

We will denote the $\gamma$-th interval of $v^i_j$ as $I^i_j(\gamma)$. 
These intervals will be referred to as the \emph{vertex intervals}.
Note that none of these intervals have common left endpoints or common right endpoints. 
Furthermore, for each of the $V_i$ we add an additional collection of $2mq+2$ intervals \[S_i := \left\{\left(\frac{q-1}{q}+\gamma \frac{1}{q} + i \epsilon,1+\gamma \frac{1}{q}+ i \epsilon\right)\colon  \gamma \in \{0, \ldots , 2mq+1\}\right\},\]
where again $\epsilon \ll \frac{1}{qk}$.

Finally, we add further intervals for each edge in $G$.
Let $e_{\gamma} = v^i_a v^j_b$ be an edge with $v^i_a \in V_i$, $v^j_b \in V_j$. W.l.o.g. we may assume that $a \leq b$ and if $a = b$, then $i < j$. We add an interval of the form $I(e_\gamma)=(r(I^i_a(\gamma)), \ell(I^j_{b}(\gamma+1))) $. 
These intervals will be referred to as the \emph{edge intervals}.
As none of the intervals of different vertices have common endpoints, we can be sure that each of these edge intervals has strictly positive length.
In the following, we will frequently identify the intervals and vertices of $G'$ in order to simplify the notation.

In the next step, we assign a label set to each of the intervals in order to construct a simultaneous interval graph. To each interval in $S_i$ we assign the label set $\{i\}$ and to each $I(e_{\gamma})$ we assign the label set $\{k+1\}$ if $\gamma$ is odd and $\{k+2\}$ if $\gamma$ is even.

Before we label the vertex intervals, we need the following observation which follows easily from the definitions above.

\begin{observation}\label{obs:twoint}
    Any interval $I^i_j(\gamma)$ intersects at most two edge intervals and these intervals have distinct labels.
\end{observation}

Any interval of a $W^i_j$ is given at least the label $i$. Let $I^i_j(\gamma)$ be one of the intervals representing the vertices of $G$. If $I^i_j(\gamma)$ does not intersect any edge interval such that one of the endpoints of that edge is contained in $V_i$, then $L(I^i_j(\gamma)) = \{i\}$. If $I^i_j(\gamma)$ intersects some edge interval such that one of that edges endpoints is contained in $V_i$ but is not identical to $v^i_j$, then we add the label of that edge to $L(I^i_j(\gamma))$. If $I^i_j(\gamma)$ intersects some edge interval such that one of that edges' endpoints is identical to $v^i_j$, then $L(I^i_j(\gamma))$ does not contain the label of that edge. Note that these last two rules cannot lead to a contradiction, due to \cref{obs:twoint}. Therefore, any interval $I^i_j(\gamma)$ has label set either $\{i\}$, $\{i,k+1\}$, $\{i, k+2\}$ or $\{i,k+1,k+2\}$.


\begin{lemma}\label{lem:intervals}
    Any minimum independent dominating set of $G'$ must contain all the vertices corresponding to the intervals in the set $W^1_{j_1} \cup \ldots \cup W^k_{j_k}$ for some set of indices $\{j_1, \ldots , j_k\}$.
\end{lemma}
\begin{proof}
    Let $D$ be some minimum independent dominating set of $G'$. For some index $i$ let $W^i_{j^*}$ be a set containing the most intervals belonging to vertices of $D$ among all $W^i_j$.
    

    Suppose there is some $\gamma$ such that $I^i_{j^*}(\gamma)$ does not belong to $D$. If there is no interval from $W^i_{j^*}$ to the left of $I^i_{j^*}(\gamma)$, then let $\lambda$ be the leftmost left endpoint of some vertex in $S_i$. If there are further intervals of $W^i_{j^*}$ in $D$ to the left of $I^i_{j^*}(\gamma)$ then let $\lambda$ be the right endpoint of the rightmost of these. We similarly define $\rho$ as either the rightmost right endpoint of a vertex in $S_i$ or the left endpoint of the leftmost interval of $W^i_{j^*}$ to the right of $I^i_{j^*}(\gamma)$. We now claim that adding all intervals of $W^i_{j^*}$ between $\lambda$ and $\rho$ and deleting from $D$ all intervals that intersect with these gives us a smaller independent dominating set.

    Either $\lambda$ or $\rho$ must be the endpoint of some $I^i_{j^*}(\delta)$, as $D$ must contain at least one interval of $W^i_{j^*}$. Without loss of generality, let this be $\lambda$. As any interval of a $W^i_j$ that intersects the interval $(\lambda, \lambda + \frac{1}{q})$ must also intersect $I^i_{j^*}(\delta)$ by definition, we can assume that $D$ must contain two intervals of $S_i$ in $(\lambda, \lambda + \frac{1}{q})$ to be an independent dominating set. Therefore, in order to dominate all intervals between $\lambda$ and $\rho$ with non intersecting intervals, the set $D$ must contain $\lceil \rho - \lambda \rceil + 1$ many intervals that cover $(\lambda, \rho)$. On the other hand, we can cover $(\lambda, \rho)$ with $\lceil \rho - \lambda  \rceil$ intervals from $W^i_{j^*}$. This is a contradiction to $D$ being a minimum independent dominating set of $G'$, proving the statement. 
\end{proof}

\begin{lemma}\label{lem:ids}
 The vertices belonging to $W:= W^1_{j_1} \cup \ldots \cup W^k_{j_k}$ form an independent dominating set of $G'$ if and only if $C:=\{v^1_{j_1}, \ldots , v^k_{j_k}\}$ is a multicolored independent set of $G$.
\end{lemma}

\begin{proof}
    It is easy to see that the intervals of some $W^i_{j_i}$ dominate all intervals of $S_i$ as well as the intervals of all other $W_j^i$. In the following we will show that the edge intervals are also covered if $C$ is an independent set in $G$. 
    To this end, let $e_{\gamma}$ be some edge of $G$ with interval $I(e_{\gamma})$. Suppose that $e_{\gamma} = v^p_x v^s_t$. As $C$ is an independent set, we can assume without loss of generality that $v^p_x \notin C$ implying that $x \neq j_p$. 
    Furthermore, we can assume that $I(e_{\gamma})$ must intersect one of $I^p_{j_p}(\gamma -1 )$, $I^p_{j_p}(\gamma)$, and $I^p_{j_p}(\gamma +1)$ and must, by definition, share a label (either $k+1$, or $k+2$) with it. As all of these intervals are contained in $W^p_{j_p}$, the set $W$ is an independent dominating set.

    For the other direction, we need to show that if $W$ is an independent dominating set of $G'$, then $C$ is an independent set of $G$. Therefore let $W$ be an independent dominating set of $G'$ and suppose that there exists some edge $e_{\gamma} = v^p_x v^s_t$ such that $v^p_x$ and $v^s_t$ are contained in $C$. Assuming that $x \leq t$ and if $x = t$, then $p < s$, we know by definition that $I(e_{\gamma}) = (r(I^p_x(\gamma)), \ell(I^s_t(\gamma + 1)))$. Also by definition, any interval of $W$ intersecting $I(e_{\gamma})$ does not share a label with it.
\end{proof}

Combining \cref{lem:ids,lem:intervals} with the fact that MISP is $\W[1]$-hard when parameterized by the solution size and \cref{thm:lower} we get the following result.

\begin{theorem}
    The IDSP is $\W[1]$-hard when parameterized by the simultaneous interval number even if a $\si(G)$-simultaneous interval representation is given. 
    Furthermore, assuming the ETH, there is no $f(\si) n^{o(\si)}$-time algorithm for the ISDP for any computable function $f$. 
\end{theorem}

Note that this reduction cannot be easily adapted to show that independent dominating set is $\W[1]$-hard when parameterized by the simultaneous interval number \emph{and} the solution size $k$, as the minimum size of an independent dominating set in $G'$ is of the order $\Omega(k m)$.

%% file: conclusion.tex
\section{Conclusion}

While we have presented some algorithmic properties for graphs of bounded simultaneous interval number, many open problems still remain. 
First and foremost is the computation of~$\si$. 
Unsurprisingly, the computation of $\si$ is \NP-hard, however, we are not aware of any results regarding the decision problem whether $\si$ is at most some fixed value $d$. 
It still remains to be seen whether there exists a computable function $f$ and an \FPT{} or an \XP{} algorithm that for a given graph $G$ and integer $d$, either correctly determines that $\si(G)>d$ or computes an $f(d)$-simultaneous interval representation of $G$. 
Such \FPT{} algorithms are known for treewidth~\cite{MR1417901,MR4399680} and cliquewidth~\cite{MR2232389}, and \XP{} algorithms are known for the tree-independence number~\cite{zbMATH06850324,DBLP:journals/corr/abs-2207-09993}. 

Furthermore, the complexity status of many important problems is still open when parameterized by $\si$, for example that of independent set and dominating set (see~\cref{table:comp}). 
Regarding the obtained \FPT{} results, it remains to be shown whether the running times are best possible. 
Especially in the case of the clique problem, there is still a large discrepancy between the achieved running time and the lower bound.

The simultaneous representation problem has also been considered for chordal graphs~\cite{jampani2009simconf}. 
This imposes the question whether similar results can be made for a \emph{simultaneous chordal number}. 
In fact, some of the results given here for the simultaneous interval number can be directly translated for the simultaneous chordal number as well. 
However, as the Dominating Set problem is \W[2]-hard for chordal graphs, the \FPT{} algorithm for that problem given in \cref{thm:fpt-dom} does not carry over.

%% file: appendix.tex



\newcommand{\triv}{trivial}
\newcommand{\trans}{tra}
\newcommand{\interval}{int}
\newcommand{\paths}{paths}
\newcommand{\cliques}{$K_n$}
\newcommand{\trees}{trees}
\newcommand{\comatch}{co-$M$}
\newcommand{\cobip}{cobip}
\newcommand{\combip}{$K_{n,n}$}
\newcommand{\splitgraphs}{split}
\newcommand{\chordal}{chord}
\newcommand{\grid}{grid}

\begin{sidewaystable}
    \centering
    \caption{Justification of \cref{fig:parameters}. The entry in row $A$ and column $B$ explains why either bounded $A$ implies bounded $B$ or why this is not the case. 
    If there is given a graph class in this entry, then this class has bounded $A$ but unbounded $B$. 
    If there is another parameter $C$ given in the entry, then the statement follows via transitivity from the statement in the entry of row $A$ and column $C$.}
    \resizebox{\textheight}{!}{
    \begin{tabular}{l c c c c c c c c c c c c c c c c c}
    \toprule
        $\to$       & pw & tw & ecc & si & p-$\alpha$ & t-$\alpha$ & thin & box & lmim & mimw & omim & sim & cw & rw & twin & track & in \\
        \midrule
        pw          & - & \triv & \paths & T\ref{thm:pw} & si & si & si & si & si & si & si & si & tw & tw & tw & si & si  \\
        tw          & p-$\alpha$ & - & p-$\alpha$ & p-$\alpha$ & \trees & \triv & lmim & \cite{chandran2008boxicity} & \trees & cw & cw & cw & \cite{MR2148860} & cw & cw & \cite{knauer2016three} & track \\
        ecc         & tw & \cliques & - & C\ref{thm:si-eccn} & si & si & si & si & si & si & si & si & \triv & cw & cw & si & si \\
        si          & twin & twin & twin & - & C\ref{corol:si-pin} & p-$\alpha$ & T\ref{thm:thin} & thin & T\ref{thm:lmim} & lmim & lmim & lmim & twin & twin & \interval  & T\ref{thm:track} & track \\
        p-$\alpha$  & twin & twin & twin & \comatch & - & \triv & \comatch & \comatch & mim & \cobip~\cite{mengel2018lower} & t-$\alpha$ & t-$\alpha$  & twin & twin & \interval & in & \splitgraphs \\
        t-$\alpha$  & twin & twin & twin & \comatch & \trees & - & \comatch & \comatch & mim & \chordal~\cite{mengel2018lower} &  & omim & twin & twin & \interval & in & \splitgraphs \\
        thin        & twin & twin & twin & t-$\alpha$ & t-$\alpha$ & \combip & - & \cite{chandran2007independent} & \cite{bonomo2019thinness} & lmim & lmim & lmim & twin & twin & \interval & in & \combip~\cite{griggs1979extremal} \\
        box         & twin & twin & twin & sim & sim & sim & sim & - & sim & sim & sim & \grid & twin & twin & \interval & in & \combip \\
        lmim        & si & twin & twin & t-$\alpha$ & t-$\alpha$ & \combip~\cite{dallard2021treewidth} & \comatch & \comatch & - & \triv & mim & mim & twin & twin & \interval & in & \combip~\cite{griggs1979extremal}\\
        mim         & lmim & twin & twin & t-$\alpha$ & t-$\alpha$ & \combip\cite{dallard2021treewidth,vatshelle2012new} & \comatch & \comatch & \trees & - & \cite{bergougnoux2023new} & omim & \interval & cw & \interval & in & \combip~\cite{griggs1979extremal} \\
        omim        & twin & twin & twin & twin & twin & twin & \comatch & \comatch & twin & twin & - & \cite{bergougnoux2023new} & twin & twin & \interval & in & \splitgraphs \\
        sim         & twin & twin & twin & twin & twin & twin & \comatch & \comatch & twin & twin & twin & - & twin & twin & \interval & in & \splitgraphs \\
        cw          & t-$\alpha$ & t-$\alpha$ & t-$\alpha$ & t-$\alpha$ & t-$\alpha$ & \combip & \comatch & \comatch & \trees & \cite{vatshelle2012new} & mim & mim & - & \cite{oum2006approx} & \cite{bonnet2022twin} & in & \combip~\cite{griggs1979extremal} \\
        rw          & t-$\alpha$ & t-$\alpha$ & t-$\alpha$ & t-$\alpha$ & t-$\alpha$ & \combip & \comatch & \comatch & \trees & \cite{vatshelle2012new} & mim & mim & \cite{oum2006approx} & - & \cite{bonnet2022twin} & in & \combip~\cite{griggs1979extremal} \\
        twin        & t-$\alpha$ & t-$\alpha$ & t-$\alpha$ & t-$\alpha$ & t-$\alpha$ & \combip & \comatch & \comatch & \trees & \grid~\cite{bonnet2022twin,vatshelle2012new} & mim & mim & rw & \grid~\cite{bonnet2022twin,jelinek2010rank-width} & - & in & \combip~\cite{griggs1979extremal} \\
        track       & sim & sim & sim & sim & sim & sim & sim & \comatch & sim & sim & sim & \grid & twin & twin & \interval & - & \triv \\
        in          & sim & sim & sim & sim & sim & sim & sim & \comatch & sim & sim & sim & \grid & twin & twin & \interval &  \triv &  - \\
        \bottomrule
    \end{tabular}
    }
    \label{tab:my_label}
\end{sidewaystable}

%% file: main.bbl
\begin{thebibliography}{10}

\bibitem{alon1986covering}
Noga Alon.
\newblock Covering graphs by the minimum number of equivalence relations.
\newblock {\em Combinatorica}, 6(3):201--206, 1986.
\newblock \href {https://doi.org/10.1007/BF02579381}
  {\path{doi:10.1007/BF02579381}}.

\bibitem{belmonte2013graph}
Rémy Belmonte and Martin Vatshelle.
\newblock Graph classes with structured neighborhoods and algorithmic
  applications.
\newblock {\em Theoretical Computer Science}, 511:54--65, 2013.
\newblock \href {https://doi.org/10.1016/j.tcs.2013.01.011}
  {\path{doi:10.1016/j.tcs.2013.01.011}}.

\bibitem{bergougnoux2023new}
Benjamin Bergougnoux, Tuukka Korhonen, and Igor Razgon.
\newblock New width parameters for independent set: One-sided-mim-width and
  neighbor-depth, 2023.
\newblock \href {https://arxiv.org/abs/2302.10643} {\path{arXiv:2302.10643}}.

\bibitem{bertossi1984dominating}
Alan~A. Bertossi.
\newblock Dominating sets for split and bipartite graphs.
\newblock {\em Information Processing Letters}, 19(1):37--40, 1984.
\newblock \href {https://doi.org/10.1016/0020-0190(84)90126-1}
  {\path{doi:10.1016/0020-0190(84)90126-1}}.

\bibitem{blaesius2016simultan}
Thomas Bl{\"a}sius and Ignaz Rutter.
\newblock Simultaneous {PQ}-ordering with applications to constrained embedding
  problems.
\newblock {\em ACM Trans. Algorithms}, 12(2):46, 2016.
\newblock Id/No 16.
\newblock \href {https://doi.org/10.1145/2738054} {\path{doi:10.1145/2738054}}.

\bibitem{MR1417901}
Hans~L. Bodlaender.
\newblock A linear-time algorithm for finding tree-decompositions of small
  treewidth.
\newblock {\em SIAM J. Comput.}, 25(6):1305--1317, 1996.
\newblock \href {https://doi.org/10.1137/S0097539793251219}
  {\path{doi:10.1137/S0097539793251219}}.

\bibitem{zbMATH01238685}
Hans~L. Bodlaender.
\newblock A partial $k$-arboretum of graphs with bounded treewidth.
\newblock {\em Theor. Comput. Sci.}, 209(1--2):1--45, 1998.
\newblock \href {https://doi.org/10.1016/S0304-3975(97)00228-4}
  {\path{doi:10.1016/S0304-3975(97)00228-4}}.

\bibitem{bok2018note}
Jan Bok and Nikola Jedličková.
\newblock A note on simultaneous representation problem for interval and
  circular-arc graphs, 2018.
\newblock \href {https://arxiv.org/abs/1811.04062} {\path{arXiv:1811.04062}}.

\bibitem{bonnet2022twin}
{\'E}douard Bonnet, Eun~Jung Kim, St{\'e}phan Thomass{\'e}, and R{\'e}mi
  Watrigant.
\newblock Twin-width. {I}: {Tractable} {FO} model checking.
\newblock {\em J. ACM}, 69(1):3:1–3:46, 2022.
\newblock \href {https://doi.org/10.1145/3486655} {\path{doi:10.1145/3486655}}.

\bibitem{bonomo2019thinness}
Flavia Bonomo and Diego {de Estrada}.
\newblock On the thinness and proper thinness of a graph.
\newblock {\em Discrete Applied Mathematics}, 261:78--92, 2019.
\newblock \href {https://doi.org/10.1016/j.dam.2018.03.072}
  {\path{doi:10.1016/j.dam.2018.03.072}}.

\bibitem{bonomo2023intersection}
Flavia Bonomo-Braberman and Gastón~Abel Brito.
\newblock Intersection models and forbidden pattern characterizations for
  2-thin and proper 2-thin graphs.
\newblock {\em Discrete Applied Mathematics}, 339:53--77, 2023.
\newblock \href {https://doi.org/10.1016/j.dam.2023.06.013}
  {\path{doi:10.1016/j.dam.2023.06.013}}.

\bibitem{bonomo2022thinness}
Flavia Bonomo-Braberman, Carolina~L. Gonzalez, Fabiano~S. Oliveira, Moysés
  S.~Sampaio Jr., and Jayme~L. Szwarcfiter.
\newblock Thinness of product graphs.
\newblock {\em Discrete Applied Mathematics}, 312:52--71, 2022.
\newblock \href {https://doi.org/10.1016/j.dam.2021.04.003}
  {\path{doi:10.1016/j.dam.2021.04.003}}.

\bibitem{booth1982dominating}
Kellogg~S. Booth and J.~Howard Johnson.
\newblock Dominating sets in chordal graphs.
\newblock {\em SIAM J. Comput.}, 11:191--199, 1982.
\newblock \href {https://doi.org/10.1137/0211015} {\path{doi:10.1137/0211015}}.

\bibitem{chandran2008boxicity}
L.~Sunil Chandran and Naveen Sivadasan.
\newblock Boxicity and treewidth.
\newblock {\em Journal of Combinatorial Theory, Series B}, 97(5):733--744,
  2007.
\newblock \href {https://doi.org/10.1016/j.jctb.2006.12.004}
  {\path{doi:10.1016/j.jctb.2006.12.004}}.

\bibitem{chandran2007independent}
Sunil Chandran, Carlo Mannino, and Gianpaolo Oriolo.
\newblock The indepedent set problem and the thinness of a graph.
\newblock Unpublished manuscript cited in \cite{bonomo2022thinness}, 2007.

\bibitem{chaplick2021topological}
Steven Chaplick, Martin T{\"o}pfer, Jan Voborn{\'{\i}}k, and Peter Zeman.
\newblock On {{\(H\)}}-topological intersection graphs.
\newblock {\em Algorithmica}, 83(11):3281--3318, 2021.
\newblock \href {https://doi.org/10.1007/s00453-021-00846-3}
  {\path{doi:10.1007/s00453-021-00846-3}}.

\bibitem{MR3864713}
Carlo Comin and Romeo Rizzi.
\newblock An improved upper bound on maximal clique listing via rectangular
  fast matrix multiplication.
\newblock {\em Algorithmica}, 80(12):3525--3562, 2018.
\newblock \href {https://doi.org/10.1007/s00453-017-0402-5}
  {\path{doi:10.1007/s00453-017-0402-5}}.

\bibitem{corneil1984clustering}
Derek~G. Corneil and Yehoshua Perl.
\newblock Clustering and domination in perfect graphs.
\newblock {\em Discrete Applied Mathematics}, 9(1):27--39, 1984.
\newblock \href {https://doi.org/10.1016/0166-218X(84)90088-X}
  {\path{doi:10.1016/0166-218X(84)90088-X}}.

\bibitem{MR2148860}
Derek~G. Corneil and Udi Rotics.
\newblock On the relationship between clique-width and treewidth.
\newblock {\em SIAM J. Comput.}, 34(4):825--847, 2005.
\newblock \href {https://doi.org/10.1137/S0097539701385351}
  {\path{doi:10.1137/S0097539701385351}}.

\bibitem{cygan2015param}
Marek Cygan, Fedor~V. Fomin, {\L}ukasz Kowalik, Daniel Lokshtanov, D{\'a}niel
  Marx, Marcin Pilipczuk, Micha{\l} Pilipczuk, and Saket Saurabh.
\newblock {\em Parameterized Algorithms}.
\newblock Cham: Springer, 2015.
\newblock \href {https://doi.org/10.1007/978-3-319-21275-3}
  {\path{doi:10.1007/978-3-319-21275-3}}.

\bibitem{dallard2021treewidth}
Cl{\'e}ment Dallard, Martin Milani{\v{c}}, and Kenny {\v{S}}torgel.
\newblock Treewidth versus clique number. {I}: {Graph} classes with a forbidden
  structure.
\newblock {\em SIAM J. Discrete Math.}, 35(4):2618--2646, 2021.
\newblock \href {https://doi.org/10.1137/20M1352119}
  {\path{doi:10.1137/20M1352119}}.

\bibitem{zbMATH07796423}
Cl{\'e}ment Dallard, Martin Milani{\v{c}}, and Kenny {\v{S}}torgel.
\newblock Treewidth versus clique number. {II}: {Tree}-independence number.
\newblock {\em J. Comb. Theory, Ser. B}, 164:404--442, 2024.
\newblock \href {https://doi.org/10.1016/j.jctb.2023.10.006}
  {\path{doi:10.1016/j.jctb.2023.10.006}}.

\bibitem{DBLP:journals/corr/abs-2207-09993}
Clément Dallard, Fedor~V. Fomin, Petr~A. Golovach, Tuukka Korhonen, and Martin
  Milanič.
\newblock Computing tree decompositions with small independence number.
\newblock In Karl Bringmann, Martin Grohe, Gabriele Puppis, and Ola Svensson,
  editors, {\em 51th International Colloquium on Automata, Languages, and
  Programming, {ICALP} 2024, July 8-12, 2024, Tallinn, Estonia}, volume 297 of
  {\em LIPIcs}. Schloss Dagstuhl - Leibniz-Zentrum f{\"{u}}r Informatik, 2024.
\newblock To appear.

\bibitem{graphclassesinterval}
H.~N. de~Ridder~et al.
\newblock Entry ``2-interval'' in {Information System on Graph Classes and
  their Inclusions (ISGCI)}.
\newblock URL: \url{https://graphclasses.org/classes/gc_40.html}.

\bibitem{graphclassestrack}
H.~N. de~Ridder~et al.
\newblock Entry ``3-track'' in {Information System on Graph Classes and their
  Inclusions (ISGCI)}.
\newblock URL: \url{https://graphclasses.org/classes/gc_1080.html}.

\bibitem{graphclasses3k1}
H.~N. de~Ridder~et al.
\newblock Entry ``$3{K}_1$-free'' in {Information System on Graph Classes and
  their Inclusions (ISGCI)}.
\newblock URL: \url{https://graphclasses.org/classes/AUTO_399.html}.

\bibitem{fellows2009parameterized}
Michael~R. Fellows, Danny Hermelin, Frances Rosamond, and St{\'e}phane
  Vialette.
\newblock On the parameterized complexity of multiple-interval graph problems.
\newblock {\em Theoretical Computer Science}, 410(1):53--61, 2009.
\newblock \href {https://doi.org/10.1016/j.tcs.2008.09.065}
  {\path{doi:10.1016/j.tcs.2008.09.065}}.

\bibitem{fomin2020tractability}
Fedor~V. Fomin, Petr~A. Golovach, and Jean-Florent Raymond.
\newblock On the tractability of optimization problems on {{\(H\)}}-graphs.
\newblock {\em Algorithmica}, 82(9):2432--2473, 2020.
\newblock \href {https://doi.org/10.1007/s00453-020-00692-9}
  {\path{doi:10.1007/s00453-020-00692-9}}.

\bibitem{francis2015maximum}
Mathew~C. Francis, Daniel Gon{\c{c}}alves, and Pascal Ochem.
\newblock The maximum clique problem in multiple interval graphs.
\newblock {\em Algorithmica}, 71:812--836, 2015.
\newblock \href {https://doi.org/10.1007/s00453-013-9828-6}
  {\path{doi:10.1007/s00453-013-9828-6}}.

\bibitem{zbMATH03214396}
Delbert~R. Fulkerson and Oliver~A. Gross.
\newblock Incidence matrices and interval graphs.
\newblock {\em Pac. J. Math.}, 15:835--855, 1965.
\newblock \href {https://doi.org/10.2140/pjm.1965.15.835}
  {\path{doi:10.2140/pjm.1965.15.835}}.

\bibitem{garey1980complexity}
Michael~R. Garey, David~S. Johnson, Gerald~L. Miller, and Christos~H.
  Papadimitriou.
\newblock The complexity of coloring circular arcs and chords.
\newblock {\em SIAM Journal on Algebraic Discrete Methods}, 1(2):216--227,
  1980.
\newblock \href {https://doi.org/10.1137/0601025} {\path{doi:10.1137/0601025}}.

\bibitem{fanica1972algorithms}
F\u{a}nic\u{a} Gavril.
\newblock Algorithms for minimum coloring, maximum clique, minimum covering by
  cliques, and maximum independent set of a chordal graph.
\newblock {\em SIAM Journal on Computing}, 1(2):180--187, 1972.
\newblock \href {https://doi.org/10.1137/0201013} {\path{doi:10.1137/0201013}}.

\bibitem{golovach2015output}
Petr~A. Golovach, Pinar Heggernes, Mamadou~Moustapha Kant{\'e}, Dieter Kratsch,
  Sigve~H. S{\ae}ther, and Yngve Villanger.
\newblock Output-polynomial enumeration on graphs of bounded (local) linear
  {MIM}-width.
\newblock In {\em Algorithms and Computation. 26th International Symposium,
  ISAAC 2015, Nagoya, Japan, December 9--11, 2015. Proceedings}, pages
  248--258. Berlin: Springer, 2015.
\newblock \href {https://doi.org/10.1007/978-3-662-48971-0_22}
  {\path{doi:10.1007/978-3-662-48971-0_22}}.

\bibitem{golumbic2000clique}
Martin~Charles Golumbic and Udi Rotics.
\newblock On the clique-width of some perfect graph classes.
\newblock {\em Int. J. Found. Comput. Sci.}, 11(3):423--443, 2000.
\newblock \href {https://doi.org/10.1142/S0129054100000260}
  {\path{doi:10.1142/S0129054100000260}}.

\bibitem{gonzales2024d-stable}
Carolina~Lucía Gonzalez and Felix Mann.
\newblock On d-stable locally checkable problems parameterized by mim-width.
\newblock {\em Discrete Applied Mathematics}, 347:1--22, 2024.
\newblock \href {https://doi.org/10.1016/j.dam.2023.12.015}
  {\path{doi:10.1016/j.dam.2023.12.015}}.

\bibitem{MR0676866}
David~A. Gregory and Norman~J. Pullman.
\newblock On a clique covering problem of {O}rlin.
\newblock {\em Discrete Math.}, 41(1):97--99, 1982.
\newblock \href {https://doi.org/10.1016/0012-365X(82)90085-1}
  {\path{doi:10.1016/0012-365X(82)90085-1}}.

\bibitem{griggs1979extremal}
Jerrold~R. Griggs.
\newblock Extremal values of the interval number of a graph, {II}.
\newblock {\em Discrete Mathematics}, 28(1):37--47, 1979.
\newblock \href {https://doi.org/10.1016/0012-365X(79)90183-3}
  {\path{doi:10.1016/0012-365X(79)90183-3}}.

\bibitem{griggs1980extremal}
Jerrold~R. Griggs and Douglas~B. West.
\newblock Extremal values of the interval number of a graph.
\newblock {\em SIAM Journal on Algebraic Discrete Methods}, 1(1):1--7, 1980.
\newblock \href {https://doi.org/10.1137/0601001} {\path{doi:10.1137/0601001}}.

\bibitem{gupta1979optimal}
Udaiprakash~I. Gupta, Der-Tsai Lee, and Joseph Y.-T. Leung.
\newblock An optimal solution for the channel-assignment problem.
\newblock {\em IEEE Transactions on Computers}, C-28(11):807--810, 1979.
\newblock \href {https://doi.org/10.1109/TC.1979.1675260}
  {\path{doi:10.1109/TC.1979.1675260}}.

\bibitem{gupta1982efficient}
Udaiprakash~I. Gupta, Der-Tsai Lee, and Joseph Y.-T. Leung.
\newblock Efficient algorithms for interval graphs and circular-arc graphs.
\newblock {\em Networks}, 12(4):459--467, 1982.
\newblock \href {https://doi.org/10.1002/net.3230120410}
  {\path{doi:10.1002/net.3230120410}}.

\bibitem{gyarfas1995multitrack}
Andr{\'a}s Gy{\'a}rf{\'a}s and Douglas~B. West.
\newblock Multitrack interval graphs.
\newblock In {\em Proceedings of the Twenty-sixth Southeastern International
  Conference on Combinatorics, Graph Theory and Computing (Boca Raton, FL,
  1995)}, volume 109 of {\em Congressus Numerantium}, pages 109--116, 1995.

\bibitem{hashimoto1971wire}
Akihiro Hashimoto and James Stevens.
\newblock Wire routing by optimizing channel assignment within large apertures.
\newblock In {\em Proceedings of the 8th Design Automation Workshop}, DAC '71,
  pages 155--169, New York, NY, USA, 1971. Association for Computing Machinery.
\newblock \href {https://doi.org/10.1145/800158.805069}
  {\path{doi:10.1145/800158.805069}}.

\bibitem{MR1820597}
Russell Impagliazzo and Ramamohan Paturi.
\newblock On the complexity of {$k$}-{SAT}.
\newblock {\em J. Comput. System Sci.}, 62(2):367--375, 2001.
\newblock Special issue on the Fourteenth Annual IEEE Conference on
  Computational Complexity (Atlanta, GA, 1999).
\newblock \href {https://doi.org/10.1006/jcss.2000.1727}
  {\path{doi:10.1006/jcss.2000.1727}}.

\bibitem{jaffke2019mimiii}
Lars Jaffke, O-joung Kwon, Torstein J.~F. Str{\o}mme, and Jan~Arne Telle.
\newblock Mim-width. {III}. {Graph} powers and generalized distance domination
  problems.
\newblock {\em Theor. Comput. Sci.}, 796:216--236, 2019.
\newblock \href {https://doi.org/10.1016/j.tcs.2019.09.012}
  {\path{doi:10.1016/j.tcs.2019.09.012}}.

\bibitem{zbMATH07191154}
Lars Jaffke, O-joung Kwon, and Jan~Arne Telle.
\newblock Mim-width. {I}. {Induced} path problems.
\newblock {\em Discrete Appl. Math.}, 278:153--168, 2020.
\newblock \href {https://doi.org/10.1016/j.dam.2019.06.026}
  {\path{doi:10.1016/j.dam.2019.06.026}}.

\bibitem{jampani2009simconf}
Krishnam~Raju Jampani and Anna Lubiw.
\newblock The simultaneous representation problem for chordal, comparability
  and permutation graphs.
\newblock In {\em Algorithms and data structures. 11th international symposium,
  WADS 2009, Banff, Canada, August 21--23, 2009. Proceedings}, pages 387--398.
  Berlin: Springer, 2009.
\newblock \href {https://doi.org/10.1007/978-3-642-03367-4_34}
  {\path{doi:10.1007/978-3-642-03367-4_34}}.

\bibitem{jampani2010siminterval}
Krishnam~Raju Jampani and Anna Lubiw.
\newblock Simultaneous interval graphs.
\newblock In {\em Algorithms and computation. 21st international symposium,
  ISAAC 2010, Jeju Island, Korea, December 15--17, 2010. Proceedings, Part I},
  pages 206--217. Berlin: Springer, 2010.
\newblock \href {https://doi.org/10.1007/978-3-642-17517-6_20}
  {\path{doi:10.1007/978-3-642-17517-6_20}}.

\bibitem{jampani2012simultaneous}
Krishnam~Raju Jampani and Anna Lubiw.
\newblock The simultaneous representation problem for chordal, comparability
  and permutation graphs.
\newblock {\em J. Graph Algorithms Appl.}, 16(2):283--315, 2012.
\newblock \href {https://doi.org/10.7155/jgaa.00259}
  {\path{doi:10.7155/jgaa.00259}}.

\bibitem{jelinek2010rank-width}
Vít Jelínek.
\newblock The rank-width of the square grid.
\newblock {\em Discrete Applied Mathematics}, 158(7):841--850, 2010.
\newblock \href {https://doi.org/10.1016/j.dam.2009.02.007}
  {\path{doi:10.1016/j.dam.2009.02.007}}.

\bibitem{jiang2010parameterized}
Minghui Jiang.
\newblock On the parameterized complexity of some optimization problems related
  to multiple-interval graphs.
\newblock {\em Theoretical Computer Science}, 411(49):4253--4262, 2010.
\newblock \href {https://doi.org/10.1016/j.tcs.2010.09.001}
  {\path{doi:10.1016/j.tcs.2010.09.001}}.

\bibitem{jiang2013recognizing}
Minghui Jiang.
\newblock Recognizing {{\(d\)}}-interval graphs and {{\(d\)}}-track interval
  graphs.
\newblock {\em Algorithmica}, 66(3):541--563, 2013.
\newblock \href {https://doi.org/10.1007/s00453-012-9651-5}
  {\path{doi:10.1007/s00453-012-9651-5}}.

\bibitem{keil1985finding}
J.~Mark Keil.
\newblock Finding {Hamiltonian} circuits in interval graphs.
\newblock {\em Inf. Process. Lett.}, 20:201--206, 1985.
\newblock \href {https://doi.org/10.1016/0020-0190(85)90050-X}
  {\path{doi:10.1016/0020-0190(85)90050-X}}.

\bibitem{kernighan1973optimum}
Brian~W. Kernighan, Daniel G.~Schweikert Schweikert, and G~Persky.
\newblock An optimum channel-routing algorithm for polycell layouts of
  integrated circuits.
\newblock In {\em Proceedings of the 10th Design Automation Workshop}, DAC '73,
  pages 50--59, 1973.
\newblock \href {https://doi.org/10.1145/62882.62886}
  {\path{doi:10.1145/62882.62886}}.

\bibitem{knauer2016three}
Kolja Knauer and Torsten Ueckerdt.
\newblock Three ways to cover a graph.
\newblock {\em Discrete Mathematics}, 339(2):745--758, 2016.
\newblock \href {https://doi.org/10.1016/j.disc.2015.10.023}
  {\path{doi:10.1016/j.disc.2015.10.023}}.

\bibitem{MR4399680}
Tuukka Korhonen.
\newblock A single-exponential time 2-approximation algorithm for treewidth.
\newblock In {\em 2021 {IEEE} 62nd {A}nnual {S}ymposium on {F}oundations of
  {C}omputer {S}cience---{FOCS} 2021}, pages 184--192. IEEE Computer Soc., Los
  Alamitos, CA, 2022.
\newblock \href {https://doi.org/10.1109/FOCS52979.2021.00026}
  {\path{doi:10.1109/FOCS52979.2021.00026}}.

\bibitem{kou1978covering}
Lawrence~T. Kou, Larry~J. Stockmeyer, and Chak-Kuen Wong.
\newblock Covering edges by cliques with regard to keyword conflicts and
  intersection graphs.
\newblock {\em Communications of the ACM}, 21(2):135--139, 1978.
\newblock \href {https://doi.org/10.1145/359340.359346}
  {\path{doi:10.1145/359340.359346}}.

\bibitem{liu2011parameterized}
Chunmei Liu and Yinglei Song.
\newblock Parameterized complexity and inapproximability of dominating set
  problem in chordal and near chordal graphs.
\newblock {\em Journal of combinatorial optimization}, 22(4):684--698, 2011.
\newblock \href {https://doi.org/10.1007/s10878-010-9317-7}
  {\path{doi:10.1007/s10878-010-9317-7}}.

\bibitem{lokshtanov2011lower}
Daniel Lokshtanov, D{\'a}niel Marx, and Saket Saurabh.
\newblock Lower bounds based on the exponential time hypothesis.
\newblock {\em Bulletin of EATCS}, 105:41--71, 2011.
\newblock URL: \url{http://eatcs.org/beatcs/index.php/beatcs/article/view/92}.

\bibitem{mannino2007stable}
Carlo Mannino, Gianpaolo Oriolo, Federico Ricci, and Sunil Chandran.
\newblock The stable set problem and the thinness of a graph.
\newblock {\em Operations Research Letters}, 35(1):1--9, 2007.
\newblock \href {https://doi.org/10.1016/j.orl.2006.01.009}
  {\path{doi:10.1016/j.orl.2006.01.009}}.

\bibitem{marx2003short}
D{\'a}niel Marx.
\newblock A short proof of the {NP}-completeness of circular arc coloring.
\newblock Unpublished manuscript, 2003.
\newblock URL: \url{https://www.cs.bme.hu/~dmarx/papers/circularNP.pdf}.

\bibitem{mengel2018lower}
Stefan Mengel.
\newblock Lower bounds on the mim-width of some graph classes.
\newblock {\em Discrete Applied Mathematics}, 248:28--32, 2018.
\newblock \href {https://doi.org/h10.1016/j.dam.2017.04.043}
  {\path{doi:h10.1016/j.dam.2017.04.043}}.

\bibitem{milanic2022tree}
Martin Milanič and Paweł Rzążewski.
\newblock Tree decompositions with bounded independence number: beyond
  independent sets, 2022.
\newblock \href {https://arxiv.org/abs/2209.12315} {\path{arXiv:2209.12315}}.

\bibitem{mohar2001face}
Bojan Mohar.
\newblock Face covers and the genus problem for apex graphs.
\newblock {\em Journal of Combinatorial Theory, Series B}, 82(1):102--117,
  2001.
\newblock \href {https://doi.org/10.1006/jctb.2000.2026}
  {\path{doi:10.1006/jctb.2000.2026}}.

\bibitem{MR2232389}
Sang-il Oum and Paul Seymour.
\newblock Approximating clique-width and branch-width.
\newblock {\em J. Combin. Theory Ser. B}, 96(4):514--528, 2006.
\newblock \href {https://doi.org/10.1016/j.jctb.2005.10.006}
  {\path{doi:10.1016/j.jctb.2005.10.006}}.

\bibitem{oum2006approx}
Sang-Il Oum and Paul Seymour.
\newblock Approximating clique-width and branch-width.
\newblock {\em J. Comb. Theory, Ser. B}, 96(4):514--528, 2006.
\newblock \href {https://doi.org/10.1016/j.jctb.2005.10.006}
  {\path{doi:10.1016/j.jctb.2005.10.006}}.

\bibitem{MR2558610}
Boram Park, Suh-Ryung Kim, and Yoshio Sano.
\newblock The competition numbers of complete multipartite graphs and mutually
  orthogonal {L}atin squares.
\newblock {\em Discrete Math.}, 309(23-24):6464--6469, 2009.
\newblock \href {https://doi.org/10.1016/j.disc.2009.06.016}
  {\path{doi:10.1016/j.disc.2009.06.016}}.

\bibitem{pietrzak2003param}
Krzysztof Pietrzak.
\newblock On the parameterized complexity of the fixed alphabet shortest common
  supersequence and longest common subsequence problems.
\newblock {\em J. Comput. Syst. Sci.}, 67(4):757--771, 2003.
\newblock \href {https://doi.org/10.1016/S0022-0000(03)00078-3}
  {\path{doi:10.1016/S0022-0000(03)00078-3}}.

\bibitem{ramalingam1988unified}
Ganesan Ramalingam and C.~Pandu Rangan.
\newblock A unified approach to domination problems on interval graphs.
\newblock {\em Information Processing Letters}, 27(5):271--274, 1988.
\newblock \href {https://doi.org/10.1016/0020-0190(88)90091-9}
  {\path{doi:10.1016/0020-0190(88)90091-9}}.

\bibitem{rutter2019simultan}
Ignaz Rutter, Darren Strash, Peter Stumpf, and Michael Vollmer.
\newblock Simultaneous representation of proper and unit interval graphs.
\newblock In {\em 27th annual European Symposium on Algorithms, ESA 2019,
  Munich/Garching, Germany, September 9--11, 2019. Proceedings}, pages
  80:1--80:15. Schloss Dagstuhl -- Leibniz-Zentrum f{\"u}r Informatik,
  Dagstuhl, 2019.
\newblock \href {https://doi.org/10.4230/LIPIcs.ESA.2019.80}
  {\path{doi:10.4230/LIPIcs.ESA.2019.80}}.

\bibitem{tsukiyama1977new}
Shuji Tsukiyama, Mikio Ide, Hiromu Ariyoshi, and Isao Shirakawa.
\newblock A new algorithm for generating all the maximal independent sets.
\newblock {\em SIAM Journal on Computing}, 6(3):505--517, 1977.
\newblock \href {https://doi.org/10.1137/0206036} {\path{doi:10.1137/0206036}}.

\bibitem{vatshelle2012new}
Martin Vatshelle.
\newblock {\em New Width Parameters of Graphs}.
\newblock PhD thesis, University of Bergen, 2012.
\newblock URL: \url{https://hdl.handle.net/1956/6166}.

\bibitem{vygen1995np}
Jens Vygen.
\newblock {NP}-completeness of some edge-disjoint paths problems.
\newblock {\em Discrete Applied Mathematics}, 61(1):83--90, 1995.
\newblock \href {https://doi.org/10.1016/0166-218X(93)E0177-Z}
  {\path{doi:10.1016/0166-218X(93)E0177-Z}}.

\bibitem{zbMATH06850324}
Nikola Yolov.
\newblock Minor-matching hypertree width.
\newblock In {\em Proceedings of the 29th annual ACM-SIAM symposium on discrete
  algorithms, SODA 2018, New Orleans, LA, USA, January 7--10, 2018}, pages
  219--233. Philadelphia, PA: Society for Industrial {and} Applied Mathematics
  (SIAM); New York, NY: Association for Computing Machinery (ACM), 2018.
\newblock \href {https://doi.org/10.1137/1.9781611975031.16}
  {\path{doi:10.1137/1.9781611975031.16}}.

\end{thebibliography}
